\documentclass[a4paper,fleqn]{cas-sc}

\usepackage{amsthm}
\usepackage{url}
\usepackage{microtype}
\usepackage{tikz}
\usetikzlibrary{calc}
\usepackage{pgfplots}
\pgfplotsset{compat=1.18}

\newenvironment{IEEEproof}{\begin{proof}}{\end{proof}}

\newtheorem{proposition}{Proposition}
\newtheorem{lemma}[proposition]{Lemma}
\newtheorem{remark}[proposition]{Remark}
\newtheorem{corollary}[proposition]{Corollary}

\newcommand{\Rp}{\mathbb{R}}
\newcommand{\Cp}{\mathbb{C}}
\newcommand{\Real}{\Re}
\newcommand{\Imag}{\Im}
\newcommand{\Mclass}{\mathcal{M}}
\newcommand{\bfx}{\tilde{\mathbf{x}}}
\newcommand{\bfs}{\tilde{\mathbf{s}}}
\newcommand{\Sig}{\boldsymbol{\Sigma}}
\newcommand{\mub}{\boldsymbol{\mu}}
\newcommand{\Gk}{\mathbf{G}_{k_1 k_2}}
\newcommand{\SDRstar}{\mathrm{SDR}^\star}
\newcommand{\mstar}{m^\star}
\DeclareMathOperator{\Var}{Var}
\DeclareMathOperator{\Ex}{\mathbb{E}}

\begin{document}
\let\WriteBookmarks\relax
\def\floatpagepagefraction{1}
\def\textpagefraction{.001}

\shorttitle{Geometric Ceilings on Time-Frequency Masking}
\shortauthors{M. Baelde}

\title[mode = title]{Geometric Ceilings on Time-Frequency Masking for Single-Channel Separation}

\author[1]{Maxime Baelde}[orcid=0000-0002-3634-1398]
\cormark[1]
\ead{maxime.baelde@ik.me}

\affiliation[1]{organization={Independent researcher}, city={Lille}, country={France}}

\cortext[1]{Corresponding author}

\begin{abstract}
Most single-channel separators estimate a source by applying a real gain to the mixture in each time-frequency bin. The optimum of that format, which the oracle masks used as bounds do not attain, is the orthogonal projection of the source onto the line spanned by the mixture, its residual set by the angle between them. Locating an estimator reduces to the block structure of a real-linear operator on stacked spectra, giving a chain of four nested classes whose three larger terms match three assumptions on the prior: zero means, circularity and absence of inter-frequency coupling. Held fixed the chain is a cascade of four orthogonal projections; refitted per frame it collapses onto its first term, attributing the whole residual to one missing real parameter per bin, the phase. When the phase posterior is symmetric about the mixture direction, the minimum mean-square estimate falls back onto the line, with gain the posterior mean of the oracle gain and excess error its variance. On MUSDB18 a posterior mean under a non-circular Gaussian-mixture prior leaves the class yet stays 11.44~dB under the per-frame ceiling, which four times as many components and 7.5x the data do not close; a closed-form gate attributes some 70\% of it, in decibels, to the predicted variance. The widest fixed class stays 6.70~dB under the same ceiling. Leaving the class and minimising squared error are conflicting requests: the barrier lies in the criterion rather than in the prior.
\end{abstract}


\maketitle

\section*{Highlights}

\begin{itemize}
\item Exact ceiling for time-frequency masking: projection onto the mixture line.
\item Four nested operator classes match three classical assumptions on the prior.
\item Adaptivity to the frame yields 6.70~dB, more than any widening of the linear class.
\item Multi-source masks: the partition is free, bounding the mask range breaks it.
\item A symmetric phase posterior pulls the MMSE estimate back onto the real line.
\end{itemize}

\section{Introduction}
\label{sec:intro}

Single-channel source separation is most often carried out in the time-frequency domain, where the estimate of a source is formed by applying a real gain to the mixture in each bin. This format covers the classical Wiener filter, the ideal ratio mask and the ideal binary mask used as oracles throughout the field, and the great majority of the mask-based separators that followed them. It is convenient, it preserves the mixture phase, and it reduces separation to the estimation of one non-negative number per bin.

The convenience carries a geometric constraint that is rarely stated. Fix one time-frequency bin, and regard the mixture value $x$ and the source value $s$ as two vectors of the plane $\Rp^2 = \Cp$. A real gain returns $mx$ with $m$ real, so whatever the estimator has learned and however large it is, the only points it can output lie on the single line $\Rp x$. The best point of that line is the orthogonal projection of $s$ onto it, and it leaves the component of $s$ orthogonal to $x$, of length $|s|\sin\theta$ with $\theta$ the angle between source and mixture (Fig.~\ref{fig:geometry}). The format is therefore a constraint whose residual is set by an angle already present in the data, and that residual is the same for every member of the class, from the crudest binary mask to a network of arbitrary size.

That observation is not what the field currently uses to bound mask-based separation. The usual reference points are particular oracle masks computed from the true sources, chiefly the ideal ratio mask and the oracle Wiener filter \cite{vincent2007oracle,erdogan2015}. Those are members of the class, short of its optimum, and on the material used below they sit several decibels under the true optimum of the class, so an estimator can beat both while remaining strictly inside the format they are supposed to represent. The question this raises is the one we address. What exactly does the real-gain format forbid, what does an estimator have to trade away in order to leave it, and does leaving it yield anything?

We take the question by the operator instead of by the estimator. Regard separation in one bin as a real-linear map applied to the stacked real and imaginary parts of the mixture. A real gain is the most constrained such map, a scalar multiple of the identity. Relaxing it in three successive steps, first to a similitude, then to an arbitrary linear map of the plane, then to a map that lets the output in one bin depend on the input in the others, yields a chain of four nested operator classes in which any estimator can be located, and the smallest class containing a given estimator is readable from the image of the unit circle under its fitted operator. The three steps are not free parameters chosen at the output: we show that they correspond exactly to three classical assumptions on the prior of the source, namely zero means, circularity of the complex law, and absence of coupling between frequencies. Dropping one assumption enlarges the containing class by one term, and reinstating all three collapses every block back to a scalar and returns the estimator to the line, where it coincides with the Wiener filter of the Gaussian separation literature.

This gives a constructive route out of the smallest class. We fit on each source a Gaussian mixture on stacked real and imaginary spectra, with non-zero means and non-circular covariances \cite[Ch.~4]{baelde2019thesis}, for which the posterior mean of the source given the mixture is available in closed form over the pairs of mixture components. Being off the line then follows from the fitted prior instead of from a design choice at the output, so whether leaving the class helps becomes a measurement. We report that measurement on MUSDB18, and it comes with a proof of what governs it: when the phase posterior is symmetric about the mixture direction, the minimum mean-square estimate falls back onto the line, with gain the posterior mean of the oracle gain and excess error the posterior variance of that gain. Leaving the class and minimising squared error are conflicting requests, which generalises to an arbitrary coupled prior a return to the mixture phase long known for the minimum mean-square short-time spectral amplitude estimate \cite{ephraim1984mmse,gerkmann2013phase}.

The contributions of this paper are the following.
\begin{enumerate}
\item An exact characterisation of the class of real-gain estimators, its optimum in closed form and its irreducible residual as an energy-weighted average of $\sin^2\theta$, at every granularity from one bin to a whole spectrum and for any number of sources, with the ordered sub-ceilings of its non-negative and bounded subclasses, which cost about one and a quarter decibels at $L = 1024$ (Section~\ref{sec:problem}).
\item A chain of four nested classes of real-linear operators on stacked spectra, with a certificate that locates an arbitrary estimator in it from its fitted operator alone, and the exact correspondence between its three larger classes and three assumptions on the prior (Section~\ref{sec:problem:classes}).
\item The separation of the two readings of that chain, degenerate above its first term when the operator is refitted per frame and a cascade of four orthogonal projections when it is held fixed, which is what attributes the whole residual of the format to the phase (Sections~\ref{sec:problem:cascade} and~\ref{sec:results:cascade}).
\item A closed-form estimator that provably lies outside the smallest class, collapses to the Wiener filter when the three assumptions are reinstated, and carries a falsifiable prediction on its behaviour away from the support of its fit, confirmed by a jump of one decibel at the frames where the assignment of components changes against 0.03~dB elsewhere (Sections~\ref{sec:model}, \ref{sec:analysis} and~\ref{sec:results:offsupport}).
\item The conflict between leaving the class and minimising squared error, proved for an arbitrary coupled prior with no parametric phase model (Proposition~\ref{prop:pull}) and measured on MUSDB18 (Sections~\ref{sec:analysis} and~\ref{sec:results}).
\item A measurement of where such an estimator stands: a deficit to the per-frame ceiling that four times as many components, 7.5x more training data and a full covariance in place of a diagonal one each leave within a decibel, and a closed-form gate that attributes some 70\% of that deficit, read in decibels, to the posterior variance of Proposition~\ref{prop:pull} (Sections~\ref{sec:results:deficit}, \ref{sec:results:cov} and~\ref{sec:results:gate}).
\end{enumerate}

\section{Related Work}
\label{sec:related}

\subsection{Oracle Masks and Their Ceilings}
\label{sec:related:oracles}

Bounding a family of estimators from above is an established benchmarking device, and the per-class oracles of Vincent, Gribonval and Plumbley \cite{vincent2007oracle} already cover single-channel time-frequency masking among other classes; the oldest member of that family, the ideal binary mask of Wang \cite{wang2005ibm}, is a member of the same class with its gain restricted to two values. The quantity we take as the ceiling is on record as well. It is the instantaneous optimal ratio mask of the speech enhancement literature, given by Liang \emph{et al.} \cite{liang2013orm} as the mask that maximises the signal-to-noise ratio, and the target of the phase-sensitive mask of Erdogan \emph{et al.} \cite{erdogan2015}, whose equation~(1) reads $a_{\mathrm{psf}} = \Real(s/y) = (|s|/|y|)\cos\theta$, derived as the minimiser of $|\hat a y - s|^2$ under the constraint $a \in \Rp$. Their Table~1 lists that constrained optimum next to its unconstrained counterpart $a_{\mathrm{icf}} = s/y$ with the annotations ``max SNR given $a \in \Rp$'' and ``max SNR given $a \in \Cp$'', so the optimality principle is already stated there alongside the formula, and their Table~2 reports on the CHiME-2 development set, averaged over the two input signal-to-noise conditions, 20.76~dB for the real-gain optimum against 19.17~dB for its truncation to $[0,1]$ and 17.29~dB for the ideal ratio mask. Oracle mask studies since have kept refining the list of candidate masks \cite{irmbest}, and the complex-ratio and training-target literature \cite{wang2014targets,williamson2016crm} selects what to regress in place of the class the regression lands in.

Two things are missing from that body of work, and supplying them is what makes the ceiling usable as an argument. First, these are lists of candidate masks compared against one another, so the reported figures bound the members that were tried and leave the class unbounded; in particular the ideal ratio mask and the oracle Wiener filter, which the field uses as the reference points a separator is asked to approach, are ordinary members several decibels below the optimum, as the CHiME-2 figures above already show. An estimator can beat both while remaining strictly inside the format they represent, so beating them establishes nothing about the format. Second, none of this work says what membership in the class depends on. Reference \cite{erdogan2015} is a training-objective study, with no source prior and no statement about the law of $s$, so there is no route from a modelling assumption to a position with respect to the class.

The systems the field actually deploys are located by the same reading, and it separates them in two. Open-Unmix \cite{stoter2019openunmix}, the reference baseline of the music separation benchmarks, regresses a source magnitude and reads the estimate out as that magnitude carried by the mixture phase, which is a non-negative real gain applied bin by bin, so it belongs to the bounded-mask subclass of Section~\ref{sec:problem} and every figure the ceiling gives for that subclass applies to it verbatim; the multichannel Wiener post-filter that the reference implementation offers on top is the one part of it that leaves the format, and it does so along the coupling axis of Section~\ref{sec:problem:cascade} instead of by modelling phase. Conv-TasNet \cite{luo2019convtasnet} and the hybrid Demucs family \cite{defossez2021hybrid} are outside the class by construction, masking in a learned or a time-domain basis, so the operator they realise on short-time Fourier coefficients is neither a real gain nor bin-wise, and the ceiling says nothing about them. A campaign such as SiSEC 2018 \cite{stoter2018sisec}, which ranks both kinds against one another, is therefore not comparing members of one class, and that is where the ceiling is useful: it quantifies what a spectrogram-masking entry leaves unexploited before any question of architecture or training volume arises, and it turns the position of a time-domain entry into something to be measured instead of assumed.

\subsection{Widely Linear Filtering and Non-Circularity}
\label{sec:related:widelylinear}

Leaving the real-gain format means letting the per-bin operator be something other than a scalar, and single-microphone enhancement has built such operators by design. Complex masking and complex ratio masking \cite{williamson2016crm} give the operator one rotation on top of its gain. Widely linear estimation in the sense of Picinbono and Chevalier \cite{picinbono1995widely}, the name for a real-linear map on $\Cp^F$ that fails to be complex-linear, was carried into noise reduction by Benesty, Chen and Huang \cite{benesty2010widely}, who derive widely linear Wiener and tradeoff filters bin by bin and quantify the gain over their strictly linear counterparts when the signal is non-circular in the sense of Neeser and Massey \cite{neeser1993proper}, the general treatment of improper complex signals and of the estimators they call for being that of Schreier and Scharf \cite{schreier2010}. Multi-frame Wiener and MVDR filters \cite{huang2012multiframe,fischer2017mfmvdr} stack several consecutive frames of one bin and filter across them, the interframe correlation vector coupling what a bin-wise operator keeps separate. The anisotropic Gaussian phase models of Magron \emph{et al.} \cite{magron2017consistent,magron2017anisotropic} reach a comparable degree of freedom while tying the anisotropy to a phase estimate in place of learning it.

In all of these the operator is prescribed from second-order statistics of the observation, so the non-circularity is an assumption in the model of the signal and, where there is coupling, it runs along the temporal axis. The question of whether reaching those degrees of freedom suffices to pass the ceiling of the constrained format is not asked, which is unsurprising: without the exact ceiling of the previous subsection there is nothing to pass.

\subsection{Generative Priors and Phase}
\label{sec:related:generative}

The Gaussian source separation literature offers the other route to the same operators, through the prior instead of through the filter. Benaroya \emph{et al.} \cite{benaroya2006} fit a Gaussian mixture on each source and show that the conditional expectation of a source given the mixture is a responsibility-weighted combination of pairwise Wiener filters, a closed-form estimator we reuse verbatim in Section~\ref{sec:model}. Their priors, and those of the factorisation-based separators that followed, are phase-invariant: the components are zero-mean and circular, and their covariances are bin-wise. Their own conclusion names phase modelling in the transform domain as the work left undone. That the discarded parameter carries quality of its own is documented independently of any separator, by the reconstruction experiments of Paliwal, W\'ojcicki and Shannon \cite{paliwal2011phase} and by the phase-aware processing literature surveyed by Mowlaee \emph{et al.} \cite{mowlaee2016phase}; recovering it from magnitudes alone, in the manner of Griffin and Lim \cite{griffin1984}, is the other classical response to the same gap and a different problem from the one treated here, being an inverse problem on one spectrum instead of an estimator of a source.

That is the gap this paper works in. The three properties a phase-invariant prior forgoes, non-zero means, non-circularity and coupling between frequencies, are exactly the three degrees of freedom the filtering literature of Section~\ref{sec:related:widelylinear} adds by hand, and exactly what separates an estimator from the class of Section~\ref{sec:related:oracles}. No existing result connects the three, so it is not known whether a prior that relinquishes those assumptions produces an estimator that leaves the class, nor whether leaving the class yields anything against a ceiling that has never been stated for the class as a whole.

\section{Problem Statement and Geometry}
\label{sec:problem}

This section states the problem the paper answers and settles it for the format the field uses. It writes that format as a set of operators and computes its best member in closed form, together with the error no member of it avoids, at every granularity from one bin to a whole spectrum. It then places an arbitrary estimator on a chain of four nested operator classes, and separates the two readings under which that chain can be given a ceiling: a per-frame oracle, under which the chain is degenerate, and a fixed operator, under which it becomes a cascade of orthogonal projections. Nothing here depends on the model of Section~\ref{sec:model}; that model is one way of leaving the smallest class of the chain, and Section~\ref{sec:analysis} shows in what sense it is the only one a prior can produce. Every quantity in this section is a property of the data or of an operator, and none is a measurement of ours; measurements are reported in Section~\ref{sec:results}.

\subsection{Notation}
\label{sec:problem:setup}

The observed mixture is a single channel $x(n) = s_1(n) + s_2(n)$. Analysis is a short-time Fourier transform with window $w$ of length $L$, hop $L/2$ and $F = L/2 + 1$ non-redundant bins, the window being a periodic Hann window throughout, which satisfies the constant-overlap-add condition at that hop. Write $\mathbf{x}_t \in \Cp^F$ for the mixture spectrum of frame $t$ and $\mathbf{s}_{i,t} \in \Cp^F$ for the spectrum of source $i$ in that frame, so that linearity of the transform gives $\mathbf{x}_t = \mathbf{s}_{1,t} + \mathbf{s}_{2,t}$ exactly and frame by frame. Every statement below is made for a single frame, the index $t$ being dropped, and $x[f]$ and $s[f]$ denote the mixture and the source of interest in bin $f$. Real and imaginary parts are stacked into $\bfx = [\Real\mathbf{x}; \Imag\mathbf{x}]$ and $\bfs = [\Real\mathbf{s}; \Imag\mathbf{s}]$, both in $\Rp^{d}$ with $d = 2F$. The stacking map is a real-linear bijection $\Cp^F \to \Rp^d$, so nothing is lost or added by working in $\Rp^d$; what it changes is the class of operators one is allowed to write down, and that is the whole subject of this paper. A single bin read through the same map is the plane $\Rp^2 = \Cp$ carrying the real inner product $\langle u,v\rangle = \Real(u\overline{v})$, in which $\theta = \angle(s,x)$ denotes the angle between source and mixture.

The problem is then the following. A separator receives $\bfx$, returns an estimate $\hat{\bfs}$ of the source of interest, and is judged by the squared error $\|\bfs - \hat{\bfs}\|^2$. It belongs to the format in use in the field when its output in each bin is a real multiple of the mixture in that bin,
\begin{equation}
\Mclass_1 = \left\{ \hat{\mathbf{s}} : \hat{s}[f] = m[f]\,x[f],\ m[f] \in \Rp \right\},
\label{eq:class}
\end{equation}
with the nested subclasses $\Mclass_{[0,1]} \subset \Mclass_{+} \subset \Mclass_1$ obtained by restricting each $m[f]$ to $[0,1]$ and to the non-negative reals. Throughout the paper ``the class'' without further qualification means $\Mclass_1$, the widest of the three, and the subclass in play is named explicitly whenever a statement concerns the bounded or the non-negative one; the subscripts of $\Mclass_1$ to $\Mclass_4$ index the steps of the chain of Section~\ref{sec:problem:classes} and are unrelated to the interval subscript of $\Mclass_{[0,1]}$. Three questions follow, and they organise the paper. What is the best member of $\Mclass_1$, and what error does no member of it avoid? What must an estimator sacrifice in order to lie outside $\Mclass_1$? And does lying outside yield anything? The first is answered in this section by a projection argument, the second by the chain of operator classes of Section~\ref{sec:problem:classes} read against the assumptions on a prior in Section~\ref{sec:analysis:outside}, the third by measurement in Section~\ref{sec:results}.

\subsection{The Best Member of $\Mclass_1$ and Its Irreducible Residual}
\label{sec:problem:line}

The first question is settled in one bin and by one projection. Membership in $\Mclass_1$ confines the estimate of that bin to the line $\Rp x$, whatever the estimator has learned and however large it is, so the best point available to it is the orthogonal projection of $s$ onto that line and what it leaves behind is the component of $s$ orthogonal to $x$.

\begin{proposition}[Exact Ceiling of the Class]
\label{prop:ceiling}
For a fixed bin with $x \neq 0$, the minimiser of $|s - m x|^2$ over $m \in \Rp$ is
\begin{equation}
\mstar = \frac{\Real\!\left(s\,\overline{x}\right)}{|x|^2},
\qquad
\min_{m \in \Rp} |s - m x|^2 = |s|^2 \sin^2\theta.
\label{eq:mstar}
\end{equation}
\end{proposition}

\begin{IEEEproof}
$|s - mx|^2 = m^2|x|^2 - 2m\Real(s\overline{x}) + |s|^2$ is a real quadratic in $m$ with positive leading coefficient; setting its derivative to zero gives $\mstar$, and substituting gives $|s|^2 - \Real(s\overline{x})^2/|x|^2 = |s|^2(1 - \cos^2\theta)$, where $\cos\theta = \Real(s\overline{x})/(|s||x|)$.
\end{IEEEproof}

Summed over bins and frames, the ceiling of the whole class expressed as a signal-to-distortion ratio is
\begin{equation}
\SDRstar = -10\log_{10} \frac{\sum_{t,f} |s|^2 \sin^2\theta}{\sum_{t,f} |s|^2},
\label{eq:sdrstar}
\end{equation}
an energy-weighted average of $\sin^2\theta$ alone.

A bin where $x = 0$ while $s \neq 0$ falls outside Proposition~\ref{prop:ceiling}, since $\mstar$ is then undefined and no member of the class can output anything but zero there. We adopt the convention that such a bin contributes its whole energy $|s|^2$ to the numerator of \eqref{eq:sdrstar} as well as to its denominator, which is the value obtained by continuity from $\sin^2\theta = 1$, and we fix $\mstar = 0$ there. The case does not arise for generic material and is recorded only so that \eqref{eq:sdrstar} is a definition and not an almost-everywhere statement.

Proposition~\ref{prop:ceiling} is the projection theorem in dimension one, and its residual is one leg of a right triangle (Fig.~\ref{fig:geometry}). Three consequences are used throughout. The residual $|s|^2\sin^2\theta$ depends on the data alone, so no amount of training, parameters or compute moves it, and the ceiling is falsifiable independently of any baseline. The statement concerns a class of estimators and not a way of building one, so it applies to discriminative and generative systems alike: the ideal ratio mask, the oracle Wiener filter, every sigmoid-output mask network and every generative separator whose final step multiplies the mixture by an estimated real ratio lies in one of the three subclasses of \eqref{eq:class}, below the same ceiling. And the ceiling is a statistic of the angular disagreement between source and mixture, reportable per corpus and per frame length with no system involved.

\begin{remark}[The Angle as Inter-Source Interference]
\label{rem:sintheta}
The angle is not a property of the source alone but a measure of its interference with what accompanies it. Write $n = x - s$ for the sum of the other sources in the bin and $\Delta\varphi = \angle s - \angle n$ for their phase difference. Then the residual leg of Proposition~\ref{prop:ceiling} is the quadrature part of the interference term, normalised by the mixture,
\begin{equation}
|s|\sin\theta = \frac{|\Imag(s\bar n)|}{|x|} = \frac{|s|\,|n|\,|\sin\Delta\varphi|}{|x|},
\label{eq:sinint}
\end{equation}
and, with $r = |n|/|s|$ the interference-to-source amplitude ratio,
\begin{equation}
\sin^2\theta = \frac{r^2\sin^2\Delta\varphi}{1 + r^2 + 2r\cos\Delta\varphi}.
\label{eq:sinratio}
\end{equation}
Geometrically $|\Imag(s\bar n)|$ is twice the area of the triangle with vertices $0$, $s$ and $x$, so \eqref{eq:sinint} reads the residual as the height of that triangle over the side carried by the mixture (Fig.~\ref{fig:interference}a). Three readings follow, and each is a statement about the material, not about an estimator.
\end{remark}

First, the residual vanishes exactly when the two contributions are collinear in the bin, in phase or in antiphase, and its maximum over the phase difference lies strictly beyond quadrature: \eqref{eq:sinratio} is zero at $\Delta\varphi = 0$ and at $\Delta\varphi = \pi$, and maximal at $\cos\Delta\varphi = -\min(r,1/r)$, which approaches quadrature as the ratio becomes extreme in either direction and antiphase as the two amplitudes approach each other (Fig.~\ref{fig:interference}b). The case $|n| = |s|$ is the single exception, and it reflects the convention already fixed above, not a discontinuity of the formula: there \eqref{eq:sinratio} reduces to $\sin^2\theta = \sin^2(\Delta\varphi/2)$, which climbs to one at antiphase, since the mixture itself vanishes, the bin carries no direction to project on, and the whole energy of the source is lost instead of none of it.

Second, the residual is bounded by the weaker of the two contributions,
\begin{equation}
|s|\sin\theta \le \min(|s|,|n|), \qquad \text{i.e.} \qquad \sin^2\theta \le \min(1,r^2),
\label{eq:sinbound}
\end{equation}
with equality at the maximising phase difference above. The proof is one line: both $0$ and $x$ lie on $\Rp x$, so the distance from $s$ to that line is at most $\min(|s - 0|, |s - x|)$. The bound is the quantitative form of a familiar fact. A bin dominated by the source of interest, $r \ll 1$, is safe for the format whatever the phases do, since the residual cannot exceed $r$ times the source amplitude; a bin where the interference matches or exceeds the source can lose the source entirely, and does so at a phase difference that moves towards quadrature as $r$ grows. The ceiling of Proposition~\ref{prop:ceiling} is therefore an energy-weighted count of how much of a mixture sits in bins of the second kind.

Third, $\sin^2\theta = 1 - \rho^2$ with $\rho = \Real(s\bar x)/(|s||x|)$ the instantaneous in-phase correlation between source and mixture, so the residual of \eqref{eq:sdrstar} is the exact instantaneous counterpart of the $(1 - \gamma^2)\sigma_s^2$ left by a Wiener filter at coherence $\gamma$, with one bin of one frame in place of an expectation. This is the sense in which the format is not merely constrained but constrained by a quantity the field already reports: the same non-circularity diagnostics measured in Section~\ref{sec:results:diagnostics}, and in particular the quadrature correlation $\Imag c_{sx}$, are what \eqref{eq:sinint} makes responsible for the residual.

The two restrictions of \eqref{eq:class} are proper ones, since $\mstar$ is neither bounded by one nor non-negative: it exceeds one whenever $|s|\cos\theta > |x|$, that is as soon as the two sources partially cancel in the bin, and it is negative whenever $\theta$ exceeds $\pi/2$. The two subclasses therefore carry ceilings of their own, and those are obtained from Proposition~\ref{prop:ceiling} without further argument, the criterion being a convex scalar quadratic whose constrained minimiser is the projection of $\mstar$ onto the admissible interval.

\begin{corollary}[Ceilings of the Subclasses]
\label{cor:subceilings}
Let $I \subseteq \Rp$ be a closed interval and $\mstar_I = \operatorname{clip}_I(\mstar)$ the point of $I$ nearest $\mstar$. For a fixed bin with $x \neq 0$,
\begin{equation}
\min_{m \in I} |s - m x|^2 = |s|^2\sin^2\theta + |x|^2 \left( \mstar - \mstar_I \right)^2 .
\label{eq:subceiling}
\end{equation}
In particular the ceilings of $\Mclass_{[0,1]}$ and $\Mclass_{+}$ are given by \eqref{eq:sdrstar} with the numerator taken from \eqref{eq:subceiling} at $I = [0,1]$ and $I = [0,\infty)$, and $\SDRstar_{[0,1]} \leq \SDRstar_{+} \leq \SDRstar$.
\end{corollary}

\begin{IEEEproof}
Completing the square in the quadratic of Proposition~\ref{prop:ceiling} gives $|s - mx|^2 = |s|^2\sin^2\theta + |x|^2(m - \mstar)^2$ for every $m$, an exact identity. The map $m \mapsto (m - \mstar)^2$ is strictly convex with minimum at $\mstar$, so its minimiser over a closed interval is the nearest point of that interval, which is $\mstar_I$, and substituting gives \eqref{eq:subceiling}. The second term is non-negative and increases when $I$ shrinks, which orders the three ceilings since $[0,1] \subset [0,\infty) \subset \Rp$.
\end{IEEEproof}

The penalty term of \eqref{eq:subceiling} deserves reading in the two cases where it fires, because each has a physical meaning. Where $\mstar$ is negative the best non-negative mask is zero, the estimate of that bin is zero and the residual is the whole energy $|s|^2$ of the source: a bin in which source and mixture disagree in phase by more than $\pi/2$ is lost entirely to any non-negative mask. Where $\mstar$ exceeds one the best bounded mask is one, the estimate is the mixture itself, and \eqref{eq:subceiling} then reads $|s - x|^2 = |s'|^2$, the energy of the \emph{other} source $s' = x - s$. The bins in which the two sources partially cancel therefore lose exactly the energy of the interference, which is the quantity a mask bounded by one is unable to exceed by construction. The three ceilings are measured side by side in Section~\ref{sec:results:ceiling}, and the gaps between them quantify what the sigmoid output layer of a mask network forfeits before any learning takes place.

Nothing in Proposition~\ref{prop:ceiling} or in Corollary~\ref{cor:subceilings} uses the number of sources, and a mixture of more than two raises a constraint the two-source case hides. When every source of $x = \sum_i s_i$ is estimated by a real mask, the masks are asked to form a partition of unity, $\sum_i m_i = 1$, since anything else either loses or duplicates part of the mixture; this is what the softmax output layer of a multi-source mask network enforces by construction. The constraint turns out to be inactive at the optimum, and to be broken by the very restrictions of \eqref{eq:class} that a network applies in order to enforce it.

\begin{corollary}[Partition of Unity at the Optimum]
\label{cor:partition}
Let $x = \sum_{i=1}^{J} s_i$ in a bin with $x \neq 0$, let $\theta_i = \angle(s_i,x)$, and estimate each source by a member of $\Mclass_1$, $\hat{s}_i = m_i x$. Minimising $|s_i - m_i x|^2$ over $m_i \in \Rp$ separately for each $i$ gives
\begin{equation}
\sum_{i=1}^{J} \mstar_i = 1,
\qquad
\min_{m_i \in \Rp} |s_i - m_i x|^2 = |s_i|^2 \sin^2\theta_i,
\label{eq:partition}
\end{equation}
so the $J$ unconstrained optima already partition the mixture and each residual is the one Proposition~\ref{prop:ceiling} gives for that source alone. The same holds when the real gain is replaced by a complex one, the class $\Mclass_2$ defined in Section~\ref{sec:problem:classes}, where the optimal complex gain on source $i$ is the exact quotient $g_i = s_i/x$ and these quotients sum to one as well, a fact restated in its own right as Proposition~\ref{prop:degenerate} below. Under a restriction of each $m_i$ to a closed interval $I \subsetneq \Rp$ the property is no longer automatic: $\sum_i \operatorname{clip}_I(\mstar_i) = 1$ holds for every bin when $J = 2$ and $I = [0,1]$, and fails in general for $I = [0,\infty)$ at $J = 2$ and for $I = [0,1]$ as soon as $J \geq 3$.
\end{corollary}

\begin{IEEEproof}
No constraint couples the $m_i$, so the joint minimiser is the vector of individual minimisers and Proposition~\ref{prop:ceiling} applies to each source with the same mixture, which gives the residuals and $\mstar_i = \Real(s_i\overline{x})/|x|^2$. Summing over $i$ and using $\sum_i s_i = x$ gives $\sum_i \mstar_i = \Real(x\overline{x})/|x|^2 = 1$, and the same computation with $\sum_i s_i/x = 1$ settles $\Mclass_2$. For $J = 2$ and $I = [0,1]$ write $\mstar_2 = 1 - \mstar_1$: if $\mstar_1 \in [0,1]$ then so is $\mstar_2$ and neither is clipped; if $\mstar_1 > 1$ then $\mstar_2 < 0$ and the clipped pair is $(1,0)$; if $\mstar_1 < 0$ the clipped pair is $(0,1)$. The sum is one in all three cases. For $I = [0,\infty)$ and $J = 2$, take $\mstar_1 < 0$, so the clipped pair is $(0, 1 - \mstar_1)$ and sums to $1 - \mstar_1 > 1$. For $I = [0,1]$ and $J = 3$, take $(\mstar_1,\mstar_2,\mstar_3) = (1.5,\,1.2,\,-1.7)$, admissible since the three sum to one, whose clipped image $(1,1,0)$ sums to two.
\end{IEEEproof}

At the optimum of the unconstrained class the partition follows from the mixing identity, so enforcing it explicitly adds nothing and the per-source ceilings of \eqref{eq:sdrstar} remain valid one source at a time. What bites is the interaction between the partition and the restrictions of \eqref{eq:class}. The two-source bounded case is benign, the clipping of one gain being compensated exactly by the clipping of the other. Non-negative gains break the partition already at two sources, and the bounded case breaks at three, in the bins where two sources satisfy $|s_j|\cos\theta_j > |x|$ and the clipping of both cannot be compensated by the third. A multi-source separator with a bounded output layer therefore faces a genuine choice in those bins, between a partition that does not sum to the mixture and a renormalisation that leaves the per-source optimum, and Corollary~\ref{cor:subceilings} quantifies the second option bin by bin. The measurements of this paper are all made with two sources, and the experimental question of a hierarchical prior over more than two sources, raised again in Section~\ref{sec:model:mixture}, is left open.

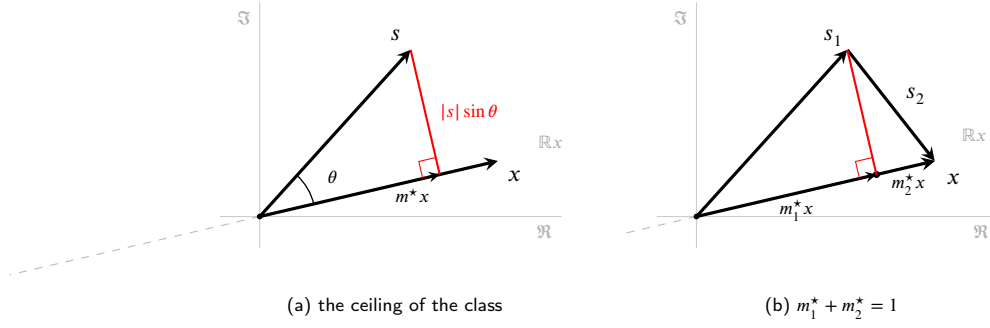
\begin{figure}[pos=tbp]
\centering
\begin{tikzpicture}[scale=2.10,>=stealth,line join=round]

\begin{scope}
  \coordinate (O) at (0,0);
  \coordinate (X) at (1.5,0.35);
  \coordinate (S) at (0.95,1.05);
  \coordinate (P) at ($(O)!0.755532!(X)$);
  \draw[gray!40] (-0.25,0) -- (1.9,0);
  \draw[gray!40] (0,-0.2) -- (0,1.35);
  \node[gray!60,anchor=north east,font=\scriptsize] at (1.9,0) {$\Re$};
  \node[gray!60,anchor=north east,font=\scriptsize] at (0,1.35) {$\Im$};
  \draw[gray!55,dashed] (O) -- ($(X)!2.05!(O)$);
  \node[gray!70,anchor=west,font=\scriptsize] at (1.7,0.47) {$\mathbb{R}x$};
  \draw[very thick,->] (O) -- (X) node[anchor=north west,font=\small] {$x$};
  \draw[very thick,->] (O) -- (S) node[anchor=south east,font=\small] {$s$};
  \draw[thick,densely dotted,->] (O) -- (P);
  \node[anchor=north east,font=\scriptsize] at ($(P)+(-0.01,-0.03)$) {$m^\star x$};
  \draw[thick,red] (P) -- (S);
  \node[red,anchor=west,font=\scriptsize] at ($(P)!0.5!(S)+(0.04,0)$) {$|s|\sin\theta$};
  \coordinate (Ra) at ($(P)!0.11cm!(S)$);
  \coordinate (Rb) at ($(P)!0.11cm!(O)$);
  \draw[red,line width=0.4pt] (Ra) -- ($(Ra)+(Rb)-(P)$) -- (Rb);
  \draw[line width=0.5pt] (0.34,0.079) arc[start angle=13.1,end angle=47.9,radius=0.35];
  \node[font=\scriptsize] at (0.46,0.25) {$\theta$};
  \fill (O) circle (0.018);
  \node[anchor=north,font=\scriptsize] at (0.85,-0.45) {(a) the ceiling of the class};
\end{scope}

\begin{scope}[shift={(2.75,0)}]
  \coordinate (O) at (0,0);
  \coordinate (X) at (1.5,0.35);
  \coordinate (S1) at (0.95,1.05);
  \coordinate (P) at ($(O)!0.755532!(X)$);
  \draw[gray!40] (-0.25,0) -- (1.9,0);
  \draw[gray!40] (0,-0.2) -- (0,1.35);
  \node[gray!60,anchor=north east,font=\scriptsize] at (1.9,0) {$\Re$};
  \node[gray!60,anchor=north east,font=\scriptsize] at (0,1.35) {$\Im$};
  \draw[gray!55,dashed] (O) -- ($(X)!1.30!(O)$);
  \node[gray!70,anchor=south east,font=\scriptsize] at (1.88,0.42) {$\mathbb{R}x$};
  \draw[very thick,->] (O) -- (S1) node[anchor=south east,font=\small,inner sep=1pt] {$s_1$};
  \draw[very thick,->] (S1) -- (X);
  \node[anchor=west,font=\small,inner sep=2pt] at (1.28,0.76) {$s_2$};
  \draw[very thick,->] (O) -- (X);
  \node[anchor=north west,font=\small,inner sep=2pt] at (1.54,0.31) {$x$};
  \draw[thick,densely dotted,->] (O) -- (P);
  \draw[thick,densely dotted,->] (P) -- (X);
  \fill (P) circle (0.022);
  \node[anchor=north,font=\scriptsize,inner sep=2pt] at ($(O)!0.42!(X)$) {$m^\star_1 x$};
  \node[anchor=north,font=\scriptsize,inner sep=2pt] at ($(P)!0.55!(X)$) {$m^\star_2 x$};
  \draw[thick,red] (P) -- (S1);
  \coordinate (Ra) at ($(P)!0.11cm!(S1)$);
  \coordinate (Rb) at ($(P)!0.11cm!(O)$);
  \draw[red,line width=0.4pt] (Ra) -- ($(Ra)+(Rb)-(P)$) -- (Rb);
  \fill (O) circle (0.018);
  \node[anchor=north,font=\scriptsize,align=center] at (0.85,-0.45) {(b) $m^\star_1 + m^\star_2 = 1$};
\end{scope}

\end{tikzpicture}
\caption{One time-frequency bin of the plane $\Rp^2 = \Cp$. (a) A real mask can only output points of the line $\Rp x$, so its best estimate is the orthogonal projection $\mstar x$ and its residual is the leg $|s|\sin\theta$ orthogonal to the mixture. (b) The two sources of a mixture project onto the same line, and since $s_1 + s_2 = x$ the two projections meet head to tail exactly at $x$: the optimal gains partition the unit, which is Corollary~\ref{cor:partition}. The single red leg is also the content of Remark~\ref{rem:shared}, the two residuals differing by a sign.}
\label{fig:geometry}
\end{figure}

\begin{figure}[pos=tbp]
\centering
\begin{tikzpicture}[scale=2.00,>=stealth,line join=round]
  \coordinate (O) at (0,0);
  \coordinate (X) at (1.90,1.07);
  \coordinate (S) at (0.55,0.95);
  \coordinate (P) at ($(O)!0.4335!(X)$);

  \draw[gray!40,->] (-0.18,0) -- (2.12,0) node[anchor=west,font=\scriptsize,black!70] {$\Re$};
  \draw[gray!40,->] (0,-0.18) -- (0,1.35) node[anchor=south,font=\scriptsize,black!70] {$\Im$};

  \fill[blue!7] (O) -- (S) -- (X) -- cycle;

  \draw[dashed,gray!55] ($(O)!-0.10!(X)$) -- ($(O)!1.14!(X)$);
  \node[anchor=south west,font=\scriptsize,gray!70] at ($(O)!1.14!(X)$) {$\Rp x$};

  \draw[very thick,->] (O) -- (X) node[anchor=north west,font=\small] {$x$};
  \draw[very thick,->] (O) -- (S) node[anchor=south east,font=\small] {$s$};
  \draw[very thick,->,black!55] (S) -- (X) node[midway,anchor=south west,font=\small] {$n$};

  \draw[very thick,red!75!black,->] (S) -- (P);
  \node[anchor=east,font=\scriptsize,red!75!black] at ($(S)!0.55!(P)$) {$|s|\sin\theta$};
  \draw[red!75!black,line width=0.35pt]
    ($(P)!0.11cm!(S)$) -- ++($($(P)!0.11cm!(O)$)-(P)$) -- ($(P)!0.11cm!(O)$);

  \draw[dashed,black!45,->] (O) -- (0.598,0.053);
  \draw[black!70,line width=0.4pt] (0.30,0.0267) arc [start angle=5.08, end angle=59.93, radius=0.301];
  \node[font=\scriptsize,anchor=west] at (0.34,0.20) {$\Delta\varphi$};

  \node[anchor=north,font=\scriptsize] at (0.95,-0.40) {(a) the residual as a triangle height};
\end{tikzpicture}
\hfill
\begin{tikzpicture}
\begin{axis}[
  width=6.4cm, height=4.6cm, scale only axis,
  xmin=0, xmax=180, ymin=0, ymax=1.08,
  xtick={0,45,90,135,180}, xticklabels={$0$,$45$,$90$,$135$,$180$},
  ytick={0,0.25,0.5,0.75,1},
  xlabel={$\Delta\varphi$ (degrees)}, ylabel={$\sin^2\theta$},
  label style={font=\scriptsize}, tick label style={font=\scriptsize},
  legend style={font=\scriptsize, at={(0.03,0.97)}, anchor=north west,
                draw=gray!50, fill=white, fill opacity=0.85, text opacity=1,
                row sep=-1pt, inner sep=2pt},
  legend cell align=left,
  axis lines=left, every axis plot/.append style={line width=0.8pt},
  clip mode=individual,
]
  \draw[dashed,gray!60] (axis cs:0,0.0625) -- (axis cs:180,0.0625);
  \draw[dashed,gray!60] (axis cs:0,0.25)   -- (axis cs:180,0.25);
  \draw[dashed,gray!60] (axis cs:0,1)      -- (axis cs:180,1);

  \addplot[black!45,   domain=0:179.4, samples=400]
    {0.0625*sin(x)^2/(1.0625+0.5*cos(x))};   \addlegendentry{$r=0.25$}
  \addplot[blue!60!black, domain=0:179.4, samples=400]
    {0.25*sin(x)^2/(1.25+cos(x))};           \addlegendentry{$r=0.5$}
  \addplot[red!75!black, domain=0:180, samples=400]
    {sin(x/2)^2};                            \addlegendentry{$r=1$}
  \addplot[teal!70!black, domain=0:179.4, samples=400]
    {4*sin(x)^2/(5+4*cos(x))};               \addlegendentry{$r=2$}
\end{axis}
\node[anchor=north,font=\scriptsize] at (3.2,-0.95) {(b) the residual against the phase difference};
\end{tikzpicture}
\caption{Physical reading of $\sin\theta$ in one bin, with $n = x - s$ the sum of the other sources and $\Delta\varphi$ their phase difference with the source of interest. (a) The residual $|s|\sin\theta$ is the height of the triangle $(0,s,x)$ over the side carried by the mixture, that is the area $\tfrac{1}{2}|\Imag(s\bar n)|$ divided by $\tfrac{1}{2}|x|$, which is \eqref{eq:sinint}. (b) The residual against $\Delta\varphi$ from \eqref{eq:sinratio}, for four amplitude ratios $r = |n|/|s|$. The dashed lines are the bound $\min(1,r^2)$ of \eqref{eq:sinbound}, reached at $\cos\Delta\varphi = -\min(r,1/r)$: a bin dominated by the source is safe whatever the phases do, a bin where the interference matches the source can lose it entirely.}
\label{fig:interference}
\end{figure}
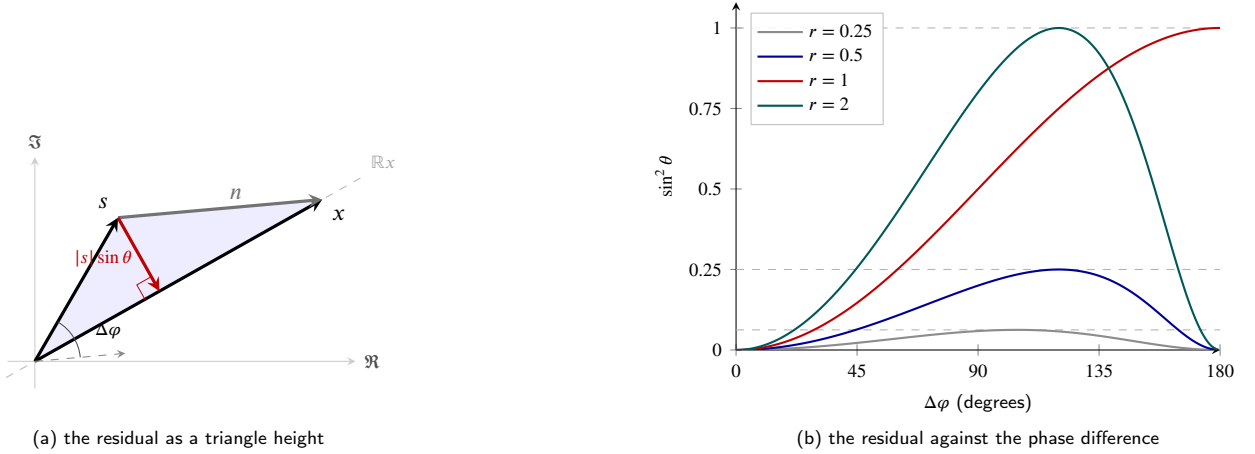

\subsection{One Gain for a Whole Spectrum}
\label{sec:problem:vector}

Proposition~\ref{prop:ceiling} is stated bin by bin because that is the format in use, but the projection it rests on does not depend on the dimension, and stating it once in $\Cp^F$ makes visible what the per-bin statement is a special case of.

\begin{lemma}[Projection Onto the Mixture Direction]
\label{lem:vector}
Let $\mathbf{s},\mathbf{x} \in \Cp^F$ with $\mathbf{x} \neq 0$, and regard $\Cp^F$ as $\Rp^{2F}$ with the real inner product $\langle \mathbf{u},\mathbf{v}\rangle = \Real(\mathbf{v}^{\mathsf{H}}\mathbf{u})$. The minimiser of $\|\mathbf{s} - m\mathbf{x}\|^2$ over $m \in \Rp$ is
\begin{equation}
M^\star = \frac{\Real\!\left(\mathbf{x}^{\mathsf{H}}\mathbf{s}\right)}{\|\mathbf{x}\|^2},
\qquad
\min_{m \in \Rp} \|\mathbf{s} - m\mathbf{x}\|^2 = \|\mathbf{s}\|^2 \sin^2\Theta,
\label{eq:vectorproj}
\end{equation}
with $\Theta$ the angle between $\mathbf{s}$ and $\mathbf{x}$ in $\Rp^{2F}$.
\end{lemma}

\begin{IEEEproof}
Identical to Proposition~\ref{prop:ceiling} with $|\cdot|$ replaced by $\|\cdot\|$ and $\Real(s\overline{x})$ by $\langle \mathbf{s},\mathbf{x}\rangle$: the objective is a real quadratic in $m$ with leading coefficient $\|\mathbf{x}\|^2 > 0$, and the residual is the squared norm of the component of $\mathbf{s}$ orthogonal to $\mathbf{x}$.
\end{IEEEproof}

One granularity is thus a special case of another, and the format sits in a hierarchy of three nested constraints, all of them projections onto a line and all of them with a residual governed by an angle: one real gain for the whole excerpt, one per frame, one per bin. Each refinement subdivides the space into more directions and can only lower the residual, so the per-bin optimum $\mstar$ is the finest member of the family and \eqref{eq:sdrstar} is the corresponding ceiling. What the comparison isolates is the contribution of the representation itself: reading $\Theta$ against the distribution of $\theta$ separates what the format loses because it is a real gain from what it recovers by being allowed to vary across bins, and that recovery is the only benefit the time-frequency representation confers on this class.

\subsection{A Hierarchy of Operator Classes}
\label{sec:problem:classes}

The second question, what an estimator must relinquish in order to lie outside $\Mclass_1$, calls for a look at operators in place of estimates. Writing the two on the same space is what makes them comparable: on the stacked vector of $\Rp^d$ a real-linear map is a matrix $\mathbf{A} \in \Rp^{d \times d}$, read as $F \times F$ blocks of size $2 \times 2$, the block $(f,g)$ acting from bin $g$ to bin $f$. Constraining that block structure gives a chain of four classes,
\begin{equation}
\Mclass_1 \subset \Mclass_2 \subset \Mclass_3 \subset \Mclass_4,
\label{eq:chain}
\end{equation}
whose signatures and geometric readings are listed in Table~\ref{tab:classes}. The smallest, $\Mclass_1$, is the real-gain format \eqref{eq:class} itself, each block a scalar multiple of the identity, hence a homothety of the plane. In $\Mclass_2$ each block is a direct similitude, the two-parameter conformal case, angle-preserving and applying one common phase rotation, which is where complex masking and complex ratio masking sit,
\begin{equation}
\Mclass_2 = \left\{ \hat{\mathbf{s}} : \hat{s}[f] = z[f]\,x[f],\ z[f] \in \Cp \right\},
\label{eq:class2}
\end{equation}
each block reading $\left[\begin{smallmatrix} a[f] & -b[f] \\ b[f] & a[f] \end{smallmatrix}\right]$ for $z[f] = a[f] + \mathrm{i}b[f]$. In $\Mclass_3$ each block is an arbitrary element of $\Rp^{2\times2}$, that is of the closure of $\mathrm{GL}(2,\Rp)$, the invertible blocks being the generic case and the singular ones kept so that the class stays convex and contains $\Mclass_2$; such a block stretches the plane by different amounts along different directions and possibly reflects it, the eccentricity of the image ellipse reading off the non-circularity of the underlying covariance; this is the widely linear case of Section~\ref{sec:related:widelylinear},
\begin{equation}
\Mclass_3 = \left\{ \hat{\mathbf{s}} : \hat{\mathbf{s}}[f] = \mathbf{A}_f\,\mathbf{x}[f],\ \mathbf{A}_f \in \Rp^{2\times2} \right\},
\label{eq:class3}
\end{equation}
with $\mathbf{x}[f] = [\Real x[f]; \Imag x[f]] \in \Rp^2$ the stacked real and imaginary parts of bin $f$, and $\hat{\mathbf{s}}[f]$ likewise, no block tying one bin to another. In $\Mclass_4$ the off-diagonal blocks are free as well, so the image in one bin depends on the content of the others,
\begin{equation}
\Mclass_4 = \left\{ \hat{\mathbf{s}} : \hat{\mathbf{s}} = \mathbf{A}\,\bfx,\ \mathbf{A} \in \Rp^{d \times d} \right\},
\label{eq:class4}
\end{equation}
the full matrix of \eqref{eq:chain}, block $(f,g)$ free to couple bin $g$ into the estimate of bin $f$. Allowing a constant offset gives the affine version of each class, which breaks homogeneity and moves the image of the unit circle off the origin.

Two properties of \eqref{eq:chain} are used below. Each inclusion is strict and each step is minimal in the sense that matters here, since it adds exactly one of the three structures that a prior can forgo, and Section~\ref{sec:analysis:outside} establishes that correspondence one to one: non-zero means, non-circularity and inter-frequency coupling. Where the chain stops is therefore decided by what a prior can forfeit instead of by the algebra of block structures, and $\Mclass_4$ is the largest class a prior of the form of Section~\ref{sec:model} can reach. The chain is not an exhaustive classification of the subclasses of $\Rp^{d \times d}$: coupled complex-linear maps, for one, lie outside it, and $\Mclass_2$ appears because complex masking occupies it.

The last column of Table~\ref{tab:classes} makes the position of an estimator readable, and Fig.~\ref{fig:unitcircle} draws it: locating an estimator means naming the smallest class of \eqref{eq:chain} that contains it, and the image of the unit circle under one fitted block decides that on its own, a concentric circle for the two masking classes according to whether it carries a common rotation, an ellipse for $\Mclass_3$, and an ellipse whose axes move with the content of another bin for $\Mclass_4$. An image displaced off the origin belongs to the affine variant of whichever class it otherwise matches. The witness is a property of the operator, so it is available before any estimate is formed, where the scalar witnesses of Section~\ref{sec:analysis:criterion} remain confounded with the conditioning of the model.

\begin{figure}[pos=tbp]
\centering
\begin{tikzpicture}[scale=1.32,>=stealth,line join=round,
  every node/.style={font=\scriptsize},
  ax/.style={gray!45,line width=0.4pt},
  ref/.style={gray!60,dashed,line width=0.5pt},
  img/.style={black,line width=0.9pt},
  ghost/.style={gray!35,line width=0.6pt},
  dir/.style={gray!70,line width=0.5pt},
  hit/.style={black,line width=0.7pt}]

\begin{scope}[shift={(0,0)}]
  \draw[ax] (-1.3,0) -- (1.3,0);
  \draw[ax] (0,-1.3) -- (0,1.3);
  \draw[ref] (0,0) circle (1);
  \draw[img] (0,0) circle (0.72);
  \draw[dir,->] (0,0) -- (0.766,0.643) node[anchor=south west,inner sep=1pt] {$u$};
  \draw[hit,->] (0,0) -- (0.5515,0.4629);
  \fill (0.5515,0.4629) circle (0.035);
  \node[anchor=west,inner sep=2pt] at (0.62,0.40) {$\mathbf{A}u$};
  \node[anchor=north] at (0,-1.5) {(a) $\Mclass_1$, real mask};
\end{scope}

\begin{scope}[shift={(3.1,0)}]
  \draw[ax] (-1.3,0) -- (1.3,0);
  \draw[ax] (0,-1.3) -- (0,1.3);
  \draw[ref] (0,0) circle (1);
  \draw[img] (0,0) circle (0.72);
  \draw[dir,->] (0,0) -- (0.766,0.643) node[anchor=south west,inner sep=1pt] {$u$};
  \draw[hit,->] (0,0) -- (0.2463,0.6766);
  \fill (0.2463,0.6766) circle (0.035);
  \node[anchor=south east,inner sep=2pt] at (0.21,0.73) {$\mathbf{A}u$};
  \draw[gray!70,line width=0.5pt,->] (0.50,0.42) arc[start angle=40,end angle=70,radius=0.655];
  \node[anchor=west,inner sep=1.5pt] at (0.56,0.60) {$\varphi$};
  \node[anchor=north] at (0,-1.5) {(b) $\Mclass_2$, complex mask};
\end{scope}

\begin{scope}[shift={(6.2,0)}]
  \draw[ax] (-1.3,0) -- (1.3,0);
  \draw[ax] (0,-1.3) -- (0,1.3);
  \draw[ref] (0,0) circle (1);
  \draw[img,cm={1.05,0.25,-0.35,0.55,(0,0)}] (0,0) circle (1);
  \draw[dir,->] (0,0) -- (0.766,0.643) node[anchor=south west,inner sep=1pt] {$u$};
  \draw[hit,->] (0,0) -- (0.5792,0.5452);
  \fill (0.5792,0.5452) circle (0.035);
  \node[anchor=west,inner sep=2pt] at (0.66,0.46) {$\mathbf{A}u$};
  \node[anchor=north] at (0,-1.5) {(c) $\Mclass_3$, widely linear};
\end{scope}

\begin{scope}[shift={(9.3,0)}]
  \draw[ax] (-1.3,0) -- (1.3,0);
  \draw[ax] (0,-1.3) -- (0,1.3);
  \draw[ref] (0,0) circle (1);
  \draw[ghost,cm={0.75,0.50,-0.20,0.95,(0,0)}] (0,0) circle (1);
  \draw[ghost,cm={1.20,-0.15,-0.55,0.40,(0,0)}] (0,0) circle (1);
  \draw[img,cm={1.05,0.25,-0.35,0.55,(0,0)}] (0,0) circle (1);
  \draw[dir,->] (0,0) -- (0.766,0.643) node[anchor=south west,inner sep=1pt] {$u$};
  \node[anchor=north] at (0,-1.5) {(d) $\Mclass_4$, coupled};
\end{scope}

\end{tikzpicture}
\caption{The witness of the chain \eqref{eq:chain}: the image (solid) of the unit circle (dashed) under one $2 \times 2$ block of each class, with one marked direction $u$ and its image $\mathbf{A}u$. In (a) the image is a concentric circle and every direction is fixed. In (b) it is again a concentric circle, but all directions rotate by one common angle $\varphi$, which is the two-parameter conformal case. In (c) it is an ellipse, so the applied rotation depends on the direction. In (d) the block is the same as in (c), drawn against two alternatives it takes when the content of the other bins changes, the coupling acting on the axes of the ellipse. An image displaced off the origin marks the affine variant.}
\label{fig:unitcircle}
\end{figure}
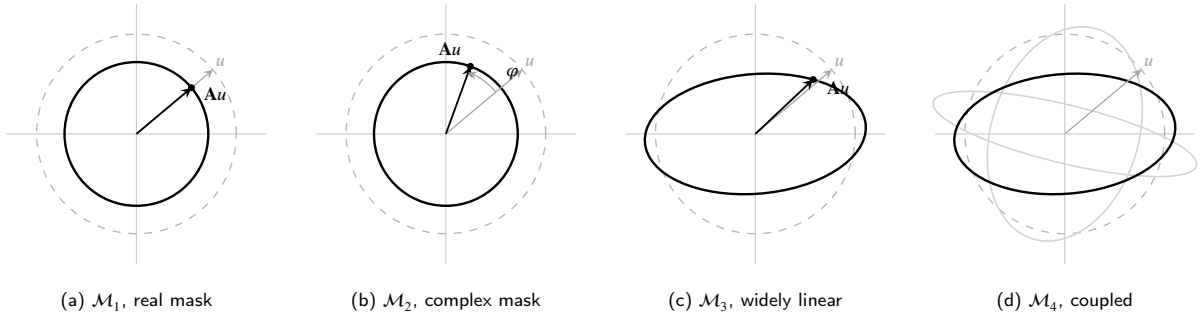

\begin{table}[tbp]
\renewcommand{\arraystretch}{1.1}
\caption{The chain $\Mclass_1 \subset \Mclass_2 \subset \Mclass_3 \subset \Mclass_4$ of real-linear operator classes on stacked spectra. Each class is reached from the model of Section~\ref{sec:model} by relaxing one assumption on the prior, and reinstating that assumption returns the estimator to the class below.}
\label{tab:classes}
\centering
\setlength{\tabcolsep}{4pt}
\begin{tabular}{p{1.7in}p{1.35in}p{1.55in}c}
\hline
\textbf{Class} & \textbf{Signature of $\mathbf{A}$} & \textbf{Image of the unit circle} & \textbf{Par./bin} \\
\hline
$\Mclass_1$, real mask & block-diag., blocks $m[f]\mathbf{I}_2$ & concentric circle, every direction fixed & 1 \\
$\Mclass_2$, complex mask & block-diag., blocks $\left[\begin{smallmatrix} a & -b \\ b & a \end{smallmatrix}\right]$ & concentric circle, one common rotation & 2 \\
$\Mclass_3$, widely linear, bin-wise & block-diag., arbitrary $2\times2$ & ellipse, rotation depends on direction & 4 \\
$\Mclass_4$, widely linear, coupled & full matrix & ellipse whose axes depend on other bins & $4F$ \\
\hline
\end{tabular}
\end{table}

\subsection{Two Notions of Ceiling and the Cascade of the Chain}
\label{sec:problem:cascade}

Asking for the best member of a class admits two readings, and they do not measure the same thing. In the first the operator is refitted in every frame with both the source and the mixture known, which is the reading Proposition~\ref{prop:ceiling} uses and the one the oracle literature of Section~\ref{sec:related:oracles} always takes. In the second a single operator is fixed once and judged in expectation over the frames of the material, which is the reading a deployed separator lives under, since no separator is handed the source it is asked to recover. The first reading is the one that gives $\Mclass_1$ a ceiling, and the first thing to say about it is that it gives the rest of the chain none.

\begin{proposition}[Degeneracy of the Per-Frame Chain]
\label{prop:degenerate}
Fix a frame in which $x[f] \neq 0$ for every bin. Then
\begin{equation}
\min_{\hat{\mathbf{s}} \in \Mclass_2} \|\mathbf{s} - \hat{\mathbf{s}}\|^2 = 0,
\label{eq:degenerate}
\end{equation}
attained bin-wise by the complex gain $g[f] = s[f]/x[f]$, and the same holds for $\Mclass_3$ and $\Mclass_4$ by inclusion.
\end{proposition}

\begin{IEEEproof}
$\Mclass_2$ allows any $g[f] \in \Cp$ in \eqref{eq:class2}, and $x[f] \neq 0$ makes $s[f]/x[f]$ admissible; substituting gives $\hat{s}[f] = s[f]$ in every bin.
\end{IEEEproof}

Read frame by frame, the whole barrier of mask-based separation lives in the single step $\Mclass_1 \to \Mclass_2$, and the residual $|s|^2\sin^2\theta$ of Proposition~\ref{prop:ceiling} is attributable in its entirety to one missing real parameter per bin, the phase. Everything above that step is free at oracle granularity, so the widely linear degrees of freedom of $\Mclass_3$ and the inter-frequency coupling of $\Mclass_4$ cannot be assessed under the first reading at all: they add exactly nothing there, the first reading having already exhausted what is available. For the upper classes of \eqref{eq:chain} to carry information the operator has to be held fixed across frames, which is the second reading, and under it the chain becomes a cascade of orthogonal projections whose steps are informative and computable in closed form.

Fix a bin and regard the estimate as an element of the real Hilbert space of square-integrable functions of the frame, with inner product $\langle u,v \rangle = \Real\,\Ex[u\overline{v}]$. The expectation is taken over frames, under the measure the measurements of Section~\ref{sec:results} use: the empirical distribution of the $T$ frames of a single excerpt, so every quantity below is a statistic of that excerpt. Averaging over a corpus instead would leave the statements unchanged, an operator held fixed over one excerpt being the more favourable of the two to the fixed-operator reading. The four classes of \eqref{eq:chain} are then four nested closed subspaces of that space, of real dimensions $1$, $2$, $4$ and $4F$ per bin, which is the last column of Table~\ref{tab:classes} read as a dimension count and not as a parameter count. The best fixed operator of each class is the orthogonal projection of $s$ onto the corresponding subspace, the four residuals $R_1 \geq R_2 \geq R_3 \geq R_4$ are the squared distances to those subspaces, and Pythagoras makes the differences telescope, each one being the squared norm of the component the wider class adds. Write $P_x = \Ex|x|^2$ for the power of the mixture in that bin, $\tilde{P}_x = \Ex[x^2]$ for its pseudo-power, $c_{sx} = \Ex[s\overline{x}]$ and $\tilde{c}_{sx} = \Ex[sx]$ for the two cross-statistics of source and mixture, and $\rho = \tilde{P}_x / P_x$ for the circularity quotient.

\begin{proposition}[The Four Fixed-Operator Residuals]
\label{prop:cascade}
Let $P_x > 0$. The residuals of the projections onto the first three classes are
\begin{equation}
R_1 = \Ex|s|^2 - \frac{(\Real c_{sx})^2}{P_x},
\qquad
R_2 = \Ex|s|^2 - \frac{|c_{sx}|^2}{P_x},
\label{eq:r12}
\end{equation}
and, when $|\rho| < 1$, $R_3 = \Ex|s|^2 - \Real(\overline{g}c_{sx} + \overline{h}\tilde{c}_{sx})$ with $(g,h)$ the solution of
\begin{equation}
\begin{pmatrix} P_x & \overline{\tilde{P}_x} \\ \tilde{P}_x & P_x \end{pmatrix}
\begin{pmatrix} g \\ h \end{pmatrix}
=
\begin{pmatrix} c_{sx} \\ \tilde{c}_{sx} \end{pmatrix}.
\label{eq:r3system}
\end{equation}
The residual of the fourth is $R_4 = \Ex|s|^2 - \mathbf{r}^{\mathsf{H}}\mathbf{R}_{zz}^{-1}\mathbf{r}$, where $\mathbf{z} = [\mathbf{x}; \overline{\mathbf{x}}] \in \Cp^{2F}$, $\mathbf{R}_{zz} = \Ex[\mathbf{z}\mathbf{z}^{\mathsf{H}}]$ and $\mathbf{r} = \Ex[\mathbf{z}\overline{s}]$. The first two steps of the cascade are
\begin{equation}
R_1 - R_2 = \frac{(\Imag c_{sx})^2}{P_x},
\qquad
R_2 - R_3 = \frac{|\tilde{c}_{sx} - \rho c_{sx}|^2}{P_x\,(1 - |\rho|^2)}.
\label{eq:gaps}
\end{equation}
\end{proposition}

\begin{IEEEproof}
Each residual is the squared distance from $s$ to a finite-dimensional subspace, so it is obtained from the normal equations of that subspace. For $\Mclass_1$ the subspace is spanned by $x$ over $\Rp$, and minimising $\Ex|s - mx|^2 = m^2 P_x - 2m\Real c_{sx} + \Ex|s|^2$ over $m \in \Rp$ gives $m = \Real c_{sx} / P_x$ and the first expression of \eqref{eq:r12}. For $\Mclass_2$ the subspace is spanned by $x$ over $\Cp$, that is by $x$ and $ix$ over $\Rp$, and the normal equation $\Ex[(s - gx)\overline{x}] = 0$ gives $g = c_{sx}/P_x$ and the second expression, whence the first step of \eqref{eq:gaps} by subtraction. For $\Mclass_3$ the subspace is spanned by $x$ and $\overline{x}$ over $\Cp$, and the two normal equations $\Ex[(s - gx - h\overline{x})\overline{x}] = 0$ and $\Ex[(s - gx - h\overline{x})x] = 0$ are \eqref{eq:r3system}, whose determinant $P_x^2 - |\tilde{P}_x|^2$ is non-zero exactly when $|\rho| < 1$; the residual of a projection being $\Ex|s|^2$ minus the inner product of $s$ with its projection gives $R_3$. Its gap to $R_2$ is read off the one-dimensional complex direction the wider class adds, obtained by orthogonalising $\overline{x}$ against $x$: the vector $e = \overline{x} - \overline{\rho}\,x$ satisfies $\Ex[e\overline{x}] = 0$ with $\Ex|e|^2 = P_x(1 - |\rho|^2)$ and $\Ex[s\overline{e}] = \tilde{c}_{sx} - \rho c_{sx}$, and the squared norm of the projection of $s$ onto it is the second step of \eqref{eq:gaps}. For $\Mclass_4$ the subspace is spanned over $\Rp$ by the $4F$ elements $\{x[g], ix[g], \overline{x[g]}, i\overline{x[g]}\}_g$, that is over $\Cp$ by the entries of $\mathbf{z}$, and the normal equations are the augmented system $\mathbf{R}_{zz}\mathbf{w} = \mathbf{r}$, whose solution gives the stated residual. Monotonicity of the four is inclusion of the subspaces.
\end{IEEEproof}

The three steps of the cascade name the three structures the chain adds, in the order Table~\ref{tab:classes} lists them, and each is a statistic of the material alone. The first step is the quadrature correlation $\Imag c_{sx}$ between source and mixture, which is the statistical trace of the phase relation. The second is the non-circularity of the mixture measured against the source through $\tilde{c}_{sx} - \rho c_{sx}$, which vanishes exactly when the augmented cross-statistics are consistent with a circular law. The third is the inter-frequency coupling, which has no scalar summary because it is the whole of the augmented Gram $\mathbf{R}_{zz}$. And the first of the three has a sign that is fixed in advance for the material this paper works on.

\begin{corollary}[The Phase Adds Nothing to a Fixed Operator]
\label{cor:nophase}
If the two sources are mutually uncorrelated in the sense $\Ex[s\overline{s'}] = 0$ with $s' = x - s$, then $c_{sx} = \Ex|s|^2$ is real and $R_1 = R_2$.
\end{corollary}

\begin{IEEEproof}
$c_{sx} = \Ex[s(\overline{s} + \overline{s'})] = \Ex|s|^2 + \Ex[s\overline{s'}] = \Ex|s|^2$, which is real and non-negative, so $\Imag c_{sx} = 0$ and the first step of \eqref{eq:gaps} vanishes.
\end{IEEEproof}

The two readings therefore disagree about the phase as sharply as they can. Per frame it accounts for everything, since it is the only parameter separating the exact ceiling of $\Mclass_1$ from exact reconstruction by Proposition~\ref{prop:degenerate}; to a fixed operator on uncorrelated sources it accounts for exactly nothing. What the phase carries is per-instance information, inaccessible to any operator held fixed across frames and therefore reachable only by an estimator that adapts to the frame it is given, which is what a prior is for and what Section~\ref{sec:model} builds. This is also what gives Proposition~\ref{prop:pull} its scope: the quadratic criterion holds the estimator in the narrowest class of the chain whatever class the prior would authorise, and Corollary~\ref{cor:nophase} says that from there the second and cheapest class is not even distinguishable from it.

\begin{remark}[Reach of the Pythagorean Reading]
\label{rem:duallyflat}
Proposition~\ref{prop:cascade} is the self-dual case of Pythagorean projection in a dually flat geometry \cite{amari2000methods,dessein2013}, the potential $\tfrac{1}{2}\|\cdot\|^2$ being its own Legendre dual and the four classes of \eqref{eq:chain} affine constraints, which is why \eqref{eq:gaps} telescopes with equality; Cardoso's decomposition of the Kullback-Leibler mismatch runs the same argument on laws and in divergence \cite{cardoso2003}, and the excess of \eqref{eq:excess} is the Bregman information of the oracle gain under the posterior \cite{banerjee2005}. The exactness does not survive an extension, the subclasses $\Mclass_{[0,1]}$ and $\Mclass_{+}$ being convex without being affine and the manifold of Gaussian mixtures with free component parameters being neither $e$- nor $m$-flat.
\end{remark}

\begin{remark}[Affine Variants]
\label{rem:affine}
Allowing a constant offset in each class amounts to projecting the centred variable $s - \Ex s$ onto the span of the centred $x - \Ex x$, so each residual of Proposition~\ref{prop:cascade} is replaced by the same expression computed on covariances in place of second moments. The two coincide when $\Ex s = \Ex x = 0$, which holds for the material of Section~\ref{sec:results} to numerical precision, so the affine variants are not reported separately. The quantity that vanishes here is the marginal mean of a bin over the frames of an excerpt, not the component means of the prior of Section~\ref{sec:model:prior}, which are non-zero by construction and are the first of the three assumptions the chain relaxes.
\end{remark}

\begin{figure}[pos=tbp]
\centering
\begin{tikzpicture}[scale=1,>=stealth,line join=round]
  \draw[gray!55] (-0.35,0) -- (-0.35,5.5);
  \foreach \v/\y in {8/0.9, 10/1.8, 12/2.7, 14/3.6, 16/4.5, 18/5.4}{
    \draw[gray!55] (-0.45,\y) -- (-0.35,\y);
    \node[gray!70,anchor=east,font=\scriptsize] at (-0.5,\y) {\v};
  }
  \node[gray!70,rotate=90,anchor=south,font=\scriptsize] at (-1.15,2.7) {SDR of the best fixed operator (dB)};
  \foreach \x/\lab in {0/{$\mathcal{M}_1$}, 1.4/{$\mathcal{M}_2$}, 2.8/{$\mathcal{M}_3$}, 4.2/{$\mathcal{M}_4$}}{
    \draw[gray!25] (\x,0) -- (\x,5.5);
    \node[anchor=north,font=\scriptsize] at (\x,-0.05) {\lab};
  }
  \node[anchor=north,font=\scriptsize,gray!70] at (0,-0.45) {1};
  \node[anchor=north,font=\scriptsize,gray!70] at (1.4,-0.45) {2};
  \node[anchor=north,font=\scriptsize,gray!70] at (2.8,-0.45) {4};
  \node[anchor=north,font=\scriptsize,gray!70] at (4.2,-0.45) {$4F$};
  \draw[dashed,gray!70] (-0.35,4.779) -- (4.6,4.779);
  \node[anchor=south east,font=\scriptsize,gray!70] at (4.6,4.779) {per-frame ceiling of $\mathcal{M}_1$, $L=1024$};
  \draw[dashed,gray!70] (-0.35,3.9735) -- (4.6,3.9735);
  \node[anchor=south east,font=\scriptsize,gray!70] at (4.6,3.9735) {per-frame ceiling of $\mathcal{M}_1$, $L=256$};
  \draw[thick] (0,0.549) -- (1.4,0.549) -- (2.8,0.558) -- (4.2,1.764);
  \foreach \x/\y in {0/0.549, 1.4/0.549, 2.8/0.558, 4.2/1.764}{
    \fill (\x,\y) circle (1.7pt);
  }
  \draw[thick,dotted] (2.8,0.558) -- (4.2,2.70);
  \draw (4.2,2.70) circle (1.7pt);
  \node[anchor=west,font=\scriptsize] at (4.32,1.764) {$L=1024$};
  \draw[thick,gray!60] (0,0.486) -- (1.4,0.486) -- (2.8,0.495) -- (4.2,1.008);
  \foreach \x/\y in {0/0.486, 1.4/0.486, 2.8/0.495, 4.2/1.008}{
    \fill[gray!60] (\x,\y) circle (1.7pt);
  }
  \draw[gray!60] (4.2,1.053) circle (1.7pt);
  \node[anchor=west,font=\scriptsize,gray!70] at (4.32,0.85) {$L=256$};
\end{tikzpicture}
\caption{The cascade of \eqref{eq:chain} under the fixed-operator reading, measured on MUSDB18 (Section~\ref{sec:results:cascade}); the row of figures under the class names is the number of real parameters per bin. Filled markers are the four residuals of Proposition~\ref{prop:cascade} as signal-to-distortion ratios, with $\Mclass_4$ corrected for in-sample optimism; the hollow marker is the uncorrected $\Mclass_4$. The dashed lines are the per-frame ceiling \eqref{eq:sdrstar} of the smallest class at the same frame length, which the whole chain of fixed operators stays several decibels below.}
\label{fig:cascade}
\end{figure}
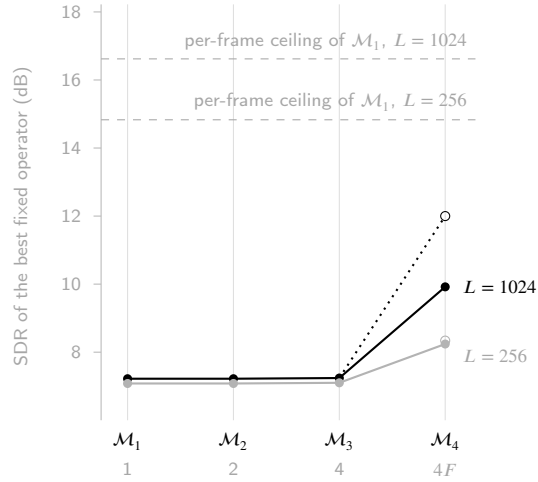

\section{Signal Model and Proposed Prior}
\label{sec:model}

Section~\ref{sec:problem} left the second question with an answer in terms of operator classes and no estimator attached to it. This section supplies one. The construction is generative: a prior is fitted on each isolated source, and the estimator is the posterior mean it induces, so where that estimator sits in the chain \eqref{eq:chain} is a consequence of the law of the prior and not a choice made at the output. The notation is that of Section~\ref{sec:problem:setup} throughout, and the construction is the closed-form Gaussian-mixture posterior of Benaroya \emph{et al.} \cite{benaroya2006} taken on stacked real and imaginary spectra instead of on power spectra, which is the single departure the whole of Section~\ref{sec:analysis} rests on.

\subsection{Prior}
\label{sec:model:prior}

Each source carries a Gaussian mixture prior on its stacked spectrum,
\begin{equation}
p_i(\bfs_i) = \sum_{k=1}^{K_i} \pi_{i,k}\,\mathcal{N}\!\left(\bfs_i;\, \mub_{i,k},\, \Sig_{i,k}\right),
\qquad \Sig_{i,k} \succ 0,
\label{eq:prior}
\end{equation}
with $K_i$ components on source $i$, written $K$ when the two counts are equal, as they are in every experiment reported here, and never to be confused with the excerpt count $n$ of the tables of Section~\ref{sec:results}. The prior is fitted by expectation-maximisation \cite{dempster1977em} on isolated source material, independently per source, following \cite[Ch.~4]{baelde2019thesis}. The components are not zero-mean, not circular in the sense of \cite{neeser1993proper}, and not bin-wise independent.

Each of those three properties is a departure from the phase-invariant priors of Section~\ref{sec:related:generative}: no constraint ties $\Sig_{rr}$ to $\Sig_{ii}$ nor forces $\Sig_{ri}$ to be antisymmetric, and the covariance is full, so frequencies are coupled. The third property has a direct physical reading, since the partials of one harmonic sound rise and fall together, so an off-diagonal block between two bins of the same series carries information about the source. Read against Table~\ref{tab:classes}, the three are exactly the three relaxations that separate $\Mclass_1$ from $\Mclass_4$, and Section~\ref{sec:analysis:outside} makes that correspondence exact. What differs from the filters of Section~\ref{sec:related:widelylinear}, which reach comparable classes by design, is the origin of the structure and the axis it runs along: there the operator is prescribed from second-order statistics of the observation and the coupling is across frames, here it is the posterior mean of a generative prior fitted on isolated sources, so the non-circularity is learned and the coupling is across frequency.

The covariance structure is therefore not an implementation detail but the knob that selects the class, and the selection makes a demand on data that has to be stated before any measurement is read. A full covariance on the stacked vector carries $d(d+1)/2$ free parameters per component against $d$ for a diagonal one, so what decides whether $\Mclass_4$ can be estimated in place of merely written down is the number of training frames per dimension available to each component, a quantity fixed by the frame length through $d = 2F$ and reported with every configuration in Section~\ref{sec:protocol:corpus}. Shortening the analysis window lowers $d$ and raises that ratio at equal frame count, against a less sparse representation and hence a different ceiling, so the two are traded against each other instead of chosen independently. The constraint is not fatal to the argument, because a diagonal covariance on $[\Real;\Imag]$ forfeits $\Mclass_4$ alone: it still allows $\Sig_{rr} \neq \Sig_{ii}$, hence a different gain on the real and on the imaginary part, hence a phase rotation that depends on the direction, which is $\Mclass_3$.

\subsection{Mixture Density}
\label{sec:model:mixture}

The observation is the sum of two independent sources, and Gaussians are closed under convolution, so the mixture inherits the form of the priors it is built from. Marginalising over the pair of active components gives a Gaussian mixture over the Cartesian product of the two component sets,
\begin{equation}
p(\bfx) = \sum_{k_1,k_2} \pi_{1,k_1}\pi_{2,k_2}\,\mathcal{N}\!\left(\bfx;\, \mub_{1,k_1} + \mub_{2,k_2},\, \Sig_{k_1k_2}\right),
\label{eq:mixdensity}
\end{equation}
with $\Sig_{k_1k_2} = \Sig_{1,k_1} + \Sig_{2,k_2}$: weights multiply, means and covariances add.

The sum therefore carries $K_1 K_2$ terms. The same closure extends the model beyond two sources: a synthetic combined source stands for any subset of them, and a hierarchical binary split then has an exact posterior at every node, since \eqref{eq:mixdensity} says the prior of a subset is again a Gaussian mixture. The hierarchy is thus an exact reparametrisation of the posterior, and the approximation enters only if one fits a mixture of $K$ components directly on a combined source instead of carrying the product of \eqref{eq:mixdensity}, at which point the error of that substitution and the influence of the order of the split become questions of their own. The geometry of the multi-source case is settled by Corollary~\ref{cor:partition}, which needs no prior; the measurements below all use two sources.

\subsection{Exact Posterior}
\label{sec:model:posterior}

Conditioning \eqref{eq:mixdensity} on the observation keeps the mixture form, each pair of components contributing one Gaussian whose mean is a Wiener correction of the component mean and whose weight is the responsibility of the pair. Writing $\Gk = \Sig_{1,k_1}\Sig_{k_1k_2}^{-1}$ for the gain of the pair,
\begin{align}
p(\bfs_1 \mid \bfx) &= \sum_{k_1,k_2} \phi_{k_1k_2}(\bfx)\,\mathcal{N}\!\left(\bfs_1;\, \tilde{\mub}_{k_1k_2}(\bfx),\, \tilde{\Sig}_{k_1k_2}\right), \label{eq:posterior}\\
\tilde{\mub}_{k_1k_2}(\bfx) &= \mub_{1,k_1} + \Gk\left(\bfx - \mub_{1,k_1} - \mub_{2,k_2}\right), \label{eq:postmean}\\
\tilde{\Sig}_{k_1k_2} &= \Sig_{1,k_1} - \Gk \Sig_{1,k_1}, \label{eq:postcov}\\
\phi_{k_1k_2}(\bfx) &\propto \pi_{1,k_1}\pi_{2,k_2}\,\mathcal{N}\!\left(\bfx;\, \mub_{1,k_1} + \mub_{2,k_2},\, \Sig_{k_1k_2}\right). \label{eq:resp}
\end{align}
The posterior covariances do not depend on the observation, only the responsibilities and the component means do.

\subsection{Estimator}
\label{sec:model:estimator}

The estimator is the mean of \eqref{eq:posterior}, which averages the pairwise Wiener estimates against the responsibilities and therefore depends on the observation twice, through the estimates and through the weights:
\begin{equation}
\hat{\mathbf{s}}_1 = \Ex\!\left[\bfs_1 \mid \bfx\right] = \sum_{k_1,k_2}\phi_{k_1k_2}(\bfx)\,\tilde{\mub}_{k_1k_2}(\bfx),
\qquad \hat{\mathbf{s}}_2 = \bfx - \hat{\mathbf{s}}_1.
\label{eq:estimator}
\end{equation}
That the conditional expectation under a Gaussian mixture prior takes the form of a responsibility-weighted combination of pairwise Wiener filters is due to Benaroya \emph{et al.} \cite[Sec.~II-B and III-B]{benaroya2006}.

\begin{remark}[The Two Sources Carry One Measurement]
\label{rem:shared}
Since $\hat{\mathbf{s}}_1 + \hat{\mathbf{s}}_2 = \bfx = \bfs_1 + \bfs_2$, the two error signals are equal up to sign and $\mathrm{SDR}_1 - \mathrm{SDR}_2 = 10\log_{10}(\|\mathbf{s}_1\|^2/\|\mathbf{s}_2\|^2)$. The same collapse holds for the ceiling, where it is a property of the class: in a bin, $\mstar_1 + \mstar_2 = 1$ by Corollary~\ref{cor:partition}, hence $s_2 - \mstar_2 x = -(s_1 - \mstar_1 x)$. Every interval reported below therefore has as many independent samples as there are test excerpts, and the agreement of the two source columns is read as a check on the implementation.
\end{remark}

The second source is recovered by subtraction instead of by a symmetric regression, so the sum constraint acts as a hard constraint and not as a penalty. Two further consequences of \eqref{eq:posterior} are used later: the posterior is a fully normalised density that can be evaluated exactly at any point (Section~\ref{sec:analysis:yardstick}), and its mean is an explicit function of $\bfx$ whose algebraic form is the object of Section~\ref{sec:analysis:outside}.

The one departure from \cite{benaroya2006} and its successors is the support of the prior. A mixture with unconstrained mean and unconstrained covariance on $[\Real;\Imag]$ is a prior on phase, learned by expectation-maximisation, with no signal model attached. Everything in Section~\ref{sec:analysis} follows from that single change.

On complexity, the gain $\Gk$ depends on the component pair and not on the frame, so the Cholesky factorisation of $\Sig_{k_1k_2}$ is computed once per pair and applied to the whole block of frames, at $O(K_1 K_2 d^3)$ once and $O(K_1 K_2 d^2)$ per frame, the quadratic form of $\phi$ reusing the same factor. The estimator is not competitive on runtime, one-step generative separation and real-time discriminative systems being both faster today, and no claim made here depends on runtime: every figure reported below is an accuracy.

\section{Theoretical Analysis}
\label{sec:analysis}

This section answers the second question of Section~\ref{sec:problem:setup} completely and settles in advance what a measurement of the third can and cannot show. It first locates the estimator of Section~\ref{sec:model} in the chain \eqref{eq:chain} and identifies, one by one, the assumptions on the prior that place it there, so that leaving $\Mclass_1$ is shown to be a property of a law, not a design choice. It then proves that the squared-error criterion works against that departure, the conditional mean returning to the line whenever the phase posterior is symmetric about the mixture direction, and identifies the resulting excess error with a posterior variance. The two remaining subsections bound the reach of those statements: how the estimator behaves away from the material the prior was fitted on, and what survives replacing the Gaussian components by an arbitrary scale mixture.

\subsection{The Class Containing the Estimator}
\label{sec:analysis:outside}

The estimator \eqref{eq:estimator} is not a linear operator, so locating it in \eqref{eq:chain} means locating the operators it is assembled from. That is what the next proposition does, and the answer is the largest class of the chain.

\begin{proposition}[Form of the Estimator]
\label{prop:form}
The estimator \eqref{eq:estimator} satisfies
\begin{equation}
\hat{\mathbf{s}}_1 = \sum_{k_1,k_2}\phi_{k_1k_2}(\bfx)\left[\Gk\,\bfx + \mathbf{c}_{k_1k_2}\right],
\label{eq:affineform}
\end{equation}
with $\mathbf{c}_{k_1k_2} = (\mathbf{I} - \Gk)\mub_{1,k_1} - \Gk\mub_{2,k_2}$: a convex combination, with observation-dependent weights, of affine widely linear operators with full inter-frequency coupling. Each term of the sum, taken at fixed responsibilities, is therefore in $\Mclass_4$, affine variant, for generic priors; the estimator itself is a mixture of such terms whose weights depend on the observation, hence not a linear operator on $\bfx$.
\end{proposition}

\begin{IEEEproof}
Substitute \eqref{eq:postmean} into \eqref{eq:estimator} and collect the terms in $\bfx$. The gain $\Gk = \Sig_{1,k_1}\Sig_{k_1k_2}^{-1}$ is a full real $2F \times 2F$ matrix, and $\mathbf{c}_{k_1k_2}$ vanishes only if both component means do.
\end{IEEEproof}

\begin{proposition}[Collapse Onto the Line]
\label{prop:collapse}
Suppose every prior component is zero-mean, bin-wise independent and circular, that is $\mub_{i,k} = \mathbf{0}$ and $\Sig_{i,k}$ block-diagonal with blocks $\tfrac{1}{2}v_{i,k}[f]\mathbf{I}_2$. Then every $\Gk$ is block-diagonal with blocks
\begin{equation}
\frac{v_{1,k_1}[f]}{v_{1,k_1}[f] + v_{2,k_2}[f]}\,\mathbf{I}_2,
\label{eq:wiener}
\end{equation}
so the estimator is a responsibility-weighted mixture of Wiener masks, $\hat{\mathbf{s}}_1 \in \Mclass_{[0,1]}$, and by Proposition~\ref{prop:ceiling} it cannot exceed $\SDRstar$.
\end{proposition}

\begin{IEEEproof}
Under the stated structure the inverse of a block-diagonal matrix with scalar blocks is block-diagonal with reciprocal scalar blocks, and the product of two such matrices is block-diagonal with the product of the scalars; a convex combination of gains in $[0,1]$ stays in $[0,1]$.
\end{IEEEproof}

The two propositions read together give the correspondence announced in Section~\ref{sec:problem:classes}, since each of the three assumptions reinstated by Proposition~\ref{prop:collapse} removes exactly one step of the chain \eqref{eq:chain}, and each is a property of the prior, not of the algorithm. Non-zero component means make the map affine instead of homogeneous, so the image circle moves off the origin. Non-circularity, that is $\Sig_{rr} \neq \Sig_{ii}$ or $\Sig_{ri}$ not antisymmetric, makes the diagonal blocks differ from $m\mathbf{I}_2$, so the circle becomes an ellipse and phase is rotated by a direction-dependent amount, which is $\Mclass_3$. A full covariance makes the off-diagonal blocks non-zero, so the ellipse in one bin depends on the content of the others, which is $\Mclass_4$. Reinstating all three collapses every block to a scalar and returns the estimator to $\Mclass_1$, where it coincides with the Wiener filter.

Recovering the Wiener reduction of \cite{benaroya2006} as the degenerate case is what makes the argument a mechanism: the estimator is confined to $\Mclass_1$ exactly under the assumptions the classical Gaussian priors make, and leaves it exactly when they are dropped. Two predictions follow and each is testable on its own. Replacing $[\Real;\Imag]$ by magnitude alone must bring the estimator back under $\mstar$, a prior on power spectra \cite{baelde2019rare} being phase-invariant; if it does not, the explanation by non-circularity is false and this analysis is wrong. Replacing the full covariance by a diagonal one on $[\Real;\Imag]$ must \emph{not} bring it back, that restriction giving up $\Mclass_4$ and keeping $\Mclass_3$. Section~\ref{sec:results} reports both restrictions: three configurations carry a diagonal covariance, so the smallest class of \eqref{eq:chain} containing their estimator is $\Mclass_3$ by construction, and one carries a full covariance, so its estimator reaches $\Mclass_4$. In neither case is it $\Mclass_1$.

\begin{remark}[Outside the Class Is Not Above the Ceiling]
\label{rem:necessary}
Being outside the class is necessary for crossing the ceiling, but it is not sufficient: an estimator outside $\Mclass_1$ can be arbitrarily poor, and most are. Proposition~\ref{prop:form} establishes the possibility, and the measurement of Section~\ref{sec:results} decides whether it is realised. The comparison is deliberately asymmetric, since $\mstar$ is an oracle computed from the true source whereas the estimator sees only the mixture and a prior fitted on other material. A failure to cross therefore indicates that the prior is not estimated well enough, and leaves the ceiling of the class untouched.
\end{remark}

\subsection{Pull of the Squared-Error Criterion}
\label{sec:analysis:criterion}

Section~\ref{sec:analysis:outside} settles which class contains the operator; this section settles where the output lands, and the two answers differ. An estimator built from a prior outside the three assumptions carries an operator in $\Mclass_4$, yet the value it returns can still lie in $\Mclass_1$, because the squared-error criterion pulls it there. Two statements carry that. The first identifies the set the criterion pulls towards, a third reading of the ceiling in which the real gain is allowed to depend on the observation and on nothing else. The second states the pull itself for an arbitrary prior and identifies the shortfall from the ceiling as a posterior variance. Together they are what makes leaving $\Mclass_1$ and minimising squared error competing demands.

\begin{lemma}[The Adaptive Real Mask]
\label{lem:adaptive}
Fix a bin with $\Ex|s|^2 < \infty$ and $x \neq 0$ almost surely, and let
\begin{equation*}
\mathcal{S} = \left\{ m(x)\,x \ : \ m : \Cp \to \Rp \ \text{measurable},\ \Ex|m(x)x|^2 < \infty \right\}
\end{equation*}
be the set of estimates that stay on the mixture line while being free to depend on the observation. Then $\mathcal{S}$ is a closed subspace of $L^2$, the orthogonal projection of $s$ onto it is
\begin{equation}
\Pi_{\mathcal{S}}\,s = \bar m\,x,
\qquad
\bar m = \Ex\!\left[\mstar \mid x\right],
\label{eq:adaptive}
\end{equation}
and the residual it leaves is
\begin{equation}
\Ex\!\left|s - \bar m x\right|^2 = \Ex\!\left[|s|^2\sin^2\theta\right] + \Ex\!\left[|x|^2\Var\!\left(\mstar \mid x\right)\right].
\label{eq:adaptiveresidual}
\end{equation}
\end{lemma}

\begin{IEEEproof}
Stability under real linear combinations is immediate, since $m_1(x)x + \lambda m_2(x)x = (m_1 + \lambda m_2)(x)\,x$ for $\lambda \in \Rp$. Closedness follows because $x \neq 0$ almost surely, so a sequence $m_n(x)x$ converging in $L^2$ to some $z$ has a subsequence converging almost surely, along which $m_n(x) \to z/x$ pointwise; the limit is therefore real, measurable with respect to $\sigma(x)$, and $z$ belongs to $\mathcal{S}$. For the projection, the completion of the square used in the proof of Corollary~\ref{cor:subceilings} is a pointwise identity, so conditioning it on $x$, under which $m(x)$ is a constant, gives
\begin{equation*}
\Ex\!\left[|s - m(x)x|^2 \mid x\right] = \Ex\!\left[|s|^2\sin^2\theta \mid x\right] + |x|^2\,\Ex\!\left[(m(x) - \mstar)^2 \mid x\right].
\end{equation*}
Only the second term depends on $m$, and for each value of $x$ it is minimised over the constant $m(x)$ at the conditional mean $\bar m = \Ex[\mstar \mid x]$, with minimum $|x|^2\Var(\mstar \mid x)$. The minimisation being pointwise in $x$, the same choice minimises the unconditional error, which gives \eqref{eq:adaptive} and \eqref{eq:adaptiveresidual}; the minimiser of a quadratic over a subspace is the orthogonal projection onto it.
\end{IEEEproof}

Writing $\Mclass_i^{\mathrm{fix}}$ for the estimates a single fixed operator of $\Mclass_i$ produces, in the sense of Section~\ref{sec:problem:cascade}, the lemma inserts $\mathcal{S}$ into the picture as follows:
\begin{equation*}
\Mclass_1^{\mathrm{fix}} \subset \Mclass_2^{\mathrm{fix}} \subset \Mclass_3^{\mathrm{fix}} \subset \Mclass_4^{\mathrm{fix}} \subset L^2\!\left(\sigma(x)\right),
\qquad
\Mclass_1^{\mathrm{fix}} \subset \mathcal{S} \subset L^2\!\left(\sigma(x)\right),
\end{equation*}
every term being a closed subspace of $L^2$ and every optimum an orthogonal projection of $s$ onto it, the largest being the conditional mean $\Ex[s \mid x]$. In both branches $x$ is the frame vector and not a single bin, the lemma being applied bin by bin so that $\mathcal{S}$ lives in the same space as the four fixed classes. The two branches meet only at the bottom and are otherwise incomparable: a fixed complex gain leaves the line without adapting to the frame, an adaptive real gain adapts without ever leaving it, and the measurement of Section~\ref{sec:results:cascade} says which of the two matters more on real material. The lemma also makes $\bar m$ the optimum of the adaptive real masks, the format every deployed mask-based separator implements, its gain being computed from the observation. Proposition~\ref{prop:pull} then reads as a coincidence condition, saying when the projection onto the whole of $L^2(\sigma(x))$ already falls inside $\mathcal{S}$.

\begin{proposition}[Pull Onto the Line]
\label{prop:pull}
Fix a bin, write $s = r\,e^{\mathrm{i}(\arg x + \theta)}$ with $r = |s|$ and $\theta = \angle(s,x)$, and let the conditional law of $(r,\theta)$ given $x$ be symmetric under $\theta \mapsto -\theta$. Then
\begin{equation}
\Ex\!\left[s \mid x\right] = \frac{\Ex[r\cos\theta \mid x]}{|x|}\,x = \Ex\!\left[\mstar \mid x\right] x,
\label{eq:pull}
\end{equation}
so the minimum mean-square estimate lies in $\Mclass_1$ exactly, its gain $\bar m = \Ex[\mstar \mid x]$ is the posterior mean of the oracle gain of Proposition~\ref{prop:ceiling}, and
\begin{equation}
\Ex\!\left[\left|s - \bar m x\right|^2 \mid x\right] = \Ex\!\left[\left|s - \mstar x\right|^2 \mid x\right] + |x|^2\Var\!\left(\mstar \mid x\right).
\label{eq:excess}
\end{equation}
Without the symmetry assumption the weaker statement still holds: $\left|\Ex[s\mid x]\right| \leq \Ex[\,|s| \mid x]$, with equality only if the posterior is supported on a single ray from the origin.
\end{proposition}

\begin{IEEEproof}
The first display is the symmetry killing the odd moment $\Ex[r\sin\theta \mid x]$, and $\mstar = (|s|/|x|)\cos\theta$ by Proposition~\ref{prop:ceiling}. The second is the bias-variance decomposition of the scalar $\mstar$ transported by the isometry $m \mapsto mx$ of $\Rp$ into $\Rp x$, using that $s - \mstar x$ is orthogonal to $x$ so the cross term vanishes. The inequality is Jensen for the norm.
\end{IEEEproof}

\begin{figure}[pos=tbp]
\centering
\begin{tikzpicture}[scale=2.30,>=stealth,line join=round,
  every node/.style={font=\scriptsize},
  lob/.style={gray!75,line width=0.6pt},
  pt/.style={gray!60,fill=gray!60},
  lead/.style={gray!60,line width=0.35pt},
  mix/.style={very thick,->}]

\begin{scope}
  \begin{scope}[rotate=13.13]
    \draw[gray!55,dashed] (-0.20,0) -- (1.92,0);
    \node[gray!70,anchor=south west,inner sep=1pt] at (1.62,0.02) {$\mathbb{R}x$};
    \draw[lob,rotate around={20:(0.95,0.44)}] (0.95,0.44) ellipse (0.30 and 0.15);
    \draw[lob,rotate around={-20:(0.95,-0.44)}] (0.95,-0.44) ellipse (0.30 and 0.15);
    \foreach \px/\py in {0.76/0.39, 0.95/0.49, 1.14/0.47, 0.76/-0.39, 0.95/-0.49, 1.14/-0.47}
      \fill[pt] (\px,\py) circle (0.020);
    \draw[gray!60,line width=0.4pt,<->] (1.34,0.34) -- (1.34,-0.34);
    \node[gray!70,anchor=north west,inner sep=1.5pt] at (1.38,0.34) {$\theta \mapsto -\theta$};
    \draw[mix] (0,0) -- (1.5403,0);
    \node[anchor=north west,inner sep=2pt,font=\small] at (1.56,-0.03) {$x$};
    \fill (0.95,0) circle (0.030);
    \draw[lead] (0.93,-0.03) -- (0.62,-0.28);
    \node[anchor=north east,inner sep=1pt] at (0.63,-0.26) {$\bar m x$};
  \end{scope}
  \fill (0,0) circle (0.016);
  \node[anchor=north,align=center] at (0.80,-0.92) {(a) symmetric phase posterior:\\ the mean stays in $\Mclass_1$};
\end{scope}

\begin{scope}[shift={(3.35,0)}]
  \begin{scope}[rotate=13.13]
    \draw[gray!55,dashed] (-0.20,0) -- (1.92,0);
    \node[gray!70,anchor=south west,inner sep=1pt] at (1.62,0.02) {$\mathbb{R}x$};
    \draw[lob,rotate around={20:(0.95,0.44)}] (0.95,0.44) ellipse (0.30 and 0.15);
    \draw[lob,gray!45,densely dotted,rotate around={-20:(0.95,-0.44)}] (0.95,-0.44) ellipse (0.17 and 0.09);
    \foreach \px/\py in {0.76/0.39, 0.95/0.49, 1.14/0.47, 0.86/0.33, 1.05/0.36}
      \fill[pt] (\px,\py) circle (0.020);
    \fill[gray!40] (0.95,-0.44) circle (0.020);
    \draw[mix] (0,0) -- (1.5403,0);
    \node[anchor=north west,inner sep=2pt,font=\small] at (1.56,-0.03) {$x$};
    \draw[red,line width=0.8pt] (0.95,0.24) -- (0.95,0);
    \fill (0.95,0.24) circle (0.030);
    \node[anchor=east,inner sep=3pt] at (0.90,0.16) {$\Ex[s \mid x]$};
  \end{scope}
  \fill (0,0) circle (0.016);
  \node[anchor=north,align=center] at (0.80,-0.92) {(b) asymmetric phase posterior:\\ the mean leaves $\Mclass_1$};
\end{scope}

\end{tikzpicture}
\caption{The mechanism of Proposition~\ref{prop:pull} in one bin, drawn in the frame where the mixture line $\Rp x$ is horizontal. Grey contours and dots stand for the posterior mass of $s$ given $x$. (a) When that mass is invariant under the reflection $\theta \mapsto -\theta$ about the mixture direction, the odd moment $\Ex[r\sin\theta \mid x]$ cancels pairwise and the barycentre is forced onto the line, so the conditional mean is a real gain however wide the class containing the operator. (b) Thinning one lobe breaks the symmetry, the barycentre lifts off the line by the red leg, and the estimate leaves $\Mclass_1$. Asymmetry of the phase posterior is thus the only route out of the class for a posterior mean.}
\label{fig:collapse}
\end{figure}
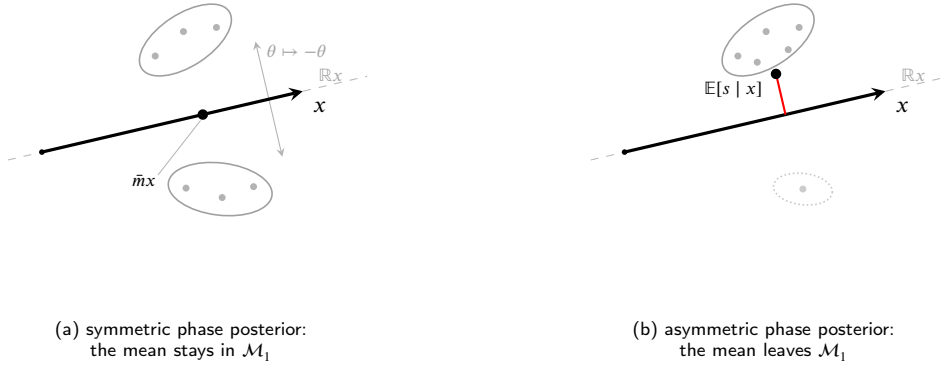

The statement has an antecedent in single-channel speech enhancement. The minimum mean-square short-time spectral amplitude estimator of Ephraim and Malah \cite{ephraim1984mmse} applies a real gain to the noisy coefficient and keeps its phase, and Gerkmann and Krawczyk \cite{gerkmann2013phase} derive the amplitude estimate that is optimal once the phase is fixed to the observed one; both estimates lie on the line $\Rp x$, and in both cases the reason is a phase model symmetric about the observed phase. Proposition~\ref{prop:pull} is that mechanism stated without a parametric model: any prior at all, coupled across frequency and with arbitrary marginals, is pulled onto the line by the conditional mean as soon as its phase posterior carries the reflection symmetry, and the excess error is identified as the posterior variance of the oracle gain by \eqref{eq:excess}. What is new here is the generality and that accounting.

Proposition~\ref{prop:pull} has three consequences. First, a posterior mean is pulled onto the line by the very criterion that defines it, and departs from the line only to the extent that the prior makes the phase posterior asymmetric about the mixture direction, which is the contrast drawn in Fig.~\ref{fig:collapse}. The two implications run in one direction each. A prior that is zero-mean, circular and bin-wise independent gives a symmetric phase posterior, hence an estimate on the line. The symmetry is immediate under those three assumptions: conditionally on the component and on $x$, the law of $s$ is complex Gaussian with mean on $\Rp x$ and circular covariance, so it is invariant under the reflection about the mixture direction, and a mixture of such laws whose weights depend on $x$ alone inherits the invariance; Proposition~\ref{prop:collapse} then places the resulting gain in $\Mclass_{[0,1]}$. Conversely, breaking that symmetry requires non-zero means or non-circularity, so those two are \emph{necessary} for a posterior mean to leave the class; they are not sufficient, since a non-circular prior can still yield a phase posterior symmetric about the mixture direction, in which case Proposition~\ref{prop:pull} returns the mean to the line even though Proposition~\ref{prop:form} places the per-pair operator in $\Mclass_4$ and no smaller class of \eqref{eq:chain}. What Proposition~\ref{prop:form} settles is the structure of the operator, what Proposition~\ref{prop:pull} settles is the position of the output, and the second is the one the ceiling responds to. Second, by \eqref{eq:excess} the shortfall from the ceiling is a posterior variance, so it shrinks as the posterior sharpens, which depends on the quality of the prior and not on the inference procedure; averaged over the observation this is exactly the second term of \eqref{eq:adaptiveresidual}, so under the symmetry the conditional mean incurs the residual of the adaptive real mask and not one decibel more. Third, the maximum a posteriori point of \eqref{eq:posterior}, and a single draw from it, have no reason to lie on the line, so they should depart from the class visibly where the mean does not, with a larger squared error. Leaving the class and minimising squared error are thus conflicting requests, and Section~\ref{sec:results:readout} measures the conflict.

\begin{remark}[Scalar Witnesses and Their Limits]
\label{rem:certificates}
Three scalar witnesses map onto the three mechanisms above, each of them zero for a member of the corresponding class: the median phase rotation between estimate and mixture, zero for a real mask; the fraction of bins whose gain has modulus above one, zero on $\Mclass_{[0,1]}$; and the fraction of bins with negative real gain, zero on $\Mclass_{+}$. A non-zero value establishes that the estimate lies off the line, and nothing more. The empirical caution is stronger: across a sweep of covariance conditioning at fixed model class, all three witnesses fall monotonically as separation quality rises, so a badly conditioned fit produces large witnesses that measure estimation noise instead of a deliberate departure from the class. We therefore read them as certificates only at the best available operating point. The geometric certificate of Section~\ref{sec:problem:classes} is free of this defect, being a property of the fitted operator that is available before any estimate is formed.
\end{remark}

\subsection{Behaviour Away From the Training Support}
\label{sec:analysis:offsupport}

The two preceding sections describe an estimator whose prior is taken as given. A prior fitted on finite material is not given everywhere, and this section asks what \eqref{eq:estimator} does when the observation falls where the fit is thin. The estimator is a convex combination of affine maps, so its sensitivity splits into a regression term and a responsibility term; the lemma below bounds the first and shows it cannot be the mechanism of failure, which leaves the second, and that is the term the training volume acts on.

\begin{lemma}[Contraction of the Regression Term]
\label{lem:contraction}
For $\Sig_1, \Sig_2 \succ 0$ and $\mathbf{S} = \Sig_1 + \Sig_2$, the matrix $\mathbf{G} = \Sig_1\mathbf{S}^{-1}$ has all its eigenvalues in $(0,1)$, and $\|\mathbf{G}\mathbf{v}\|_{\mathbf{S}^{-1}} \leq \|\mathbf{v}\|_{\mathbf{S}^{-1}}$ for $\|\mathbf{v}\|_{\mathbf{S}^{-1}} = \|\mathbf{S}^{-1/2}\mathbf{v}\|_2$. In the Euclidean norm only the weaker bound $\|\mathbf{G}\|_2 \leq \sqrt{\kappa(\mathbf{S})}$ holds, and it cannot be tightened to one.
\end{lemma}

\begin{IEEEproof}
$\mathbf{S}^{-1/2}\mathbf{G}\mathbf{S}^{1/2} = \mathbf{S}^{-1/2}\Sig_1\mathbf{S}^{-1/2} =: \mathbf{M}$ is symmetric with $\mathbf{0} \prec \mathbf{M} \prec \mathbf{I}$, since $\Sig_1 \prec \mathbf{S}$. Similar matrices share their spectrum, and $\|\mathbf{G}\mathbf{v}\|_{\mathbf{S}^{-1}} = \|\mathbf{M}\mathbf{S}^{-1/2}\mathbf{v}\|_2 \leq \|\mathbf{S}^{-1/2}\mathbf{v}\|_2$. The Euclidean bound follows from $\mathbf{G} = \mathbf{S}^{1/2}\mathbf{M}\mathbf{S}^{-1/2}$ and $\|\mathbf{M}\|_2 \leq 1$. For the counter-example, $\Sig_1 = \mathrm{diag}(1, 10^{-2})$ with $\mathbf{S} = \left[\begin{smallmatrix} 1.01 & 0.09 \\ 0.09 & 1\end{smallmatrix}\right]$, which leaves $\Sig_2 = \mathbf{S} - \Sig_1 \succ 0$, gives eigenvalues $0.998$ and $0.010$ but $\|\mathbf{G}\|_2 = 1.002$.
\end{IEEEproof}

Two observations therefore give estimates that differ by at most
\begin{equation}
\begin{split}
\left\|\hat{\mathbf{s}}_1(\bfx) - \hat{\mathbf{s}}_1(\bfx')\right\| \leq{}& \max_{k_1k_2}\left\|\Gk\right\| \left\|\bfx - \bfx'\right\| \\
&+ \max_{k_1k_2}\left\|\tilde{\mub}_{k_1k_2}(\bfx')\right\| \left\|\boldsymbol{\phi}(\bfx) - \boldsymbol{\phi}(\bfx')\right\|_1,
\end{split}
\label{eq:twoterm}
\end{equation}
a regression term, non-expansive for each component pair in the metric that pair's own covariance induces, by Lemma~\ref{lem:contraction}, and bounded in the Euclidean one by $\sqrt{\kappa(\Sig_{k_1k_2})}$, and a responsibility term.

The estimator degrades gracefully instead of catastrophically when the observation moves away from where the prior was fitted, and \eqref{eq:twoterm} says where the residual risk sits. The responsibility term is where off-support behaviour actually lives, and the geometry says why: $\boldsymbol{\phi}$ is a softmax of quadratic forms, so it partitions observation space into cells bounded by quadric surfaces, one cell per dominant component pair. Inside a cell the estimator is close to a single affine operator and the bound is governed by Lemma~\ref{lem:contraction}; crossing a wall swaps one affine operator for another whose offset $\mathbf{c}_{k_1k_2}$ may be far away.

This yields a falsifiable prediction. Under a continuously increasing deformation of the observed spectrum, the error should move slowly while the responsibilities are stable, and break where the arg max of $\boldsymbol{\phi}$ switches; Section~\ref{sec:results:offsupport} runs that test. A smooth curve with no break would show that the responsibility term is not the mechanism, and would refute the present analysis. The same reading explains why the off-support regime is data-hungry, not broken: what fails first is the assignment and not the regression, and the assignment improves with the number of cells correctly populated, hence with the amount of training material. Section~\ref{sec:results:deficit} measures that data limitation directly, at a fixed model class.

\subsection{The Exact Posterior as a Reference for Samplers}
\label{sec:analysis:yardstick}

Contemporary generative separators define a posterior implicitly and reach it by sampling, whether by annealed Langevin dynamics on a prior fitted per source \cite{jayaram2020basis}, by score matching of a mixing diffusion \cite{scheibler2023diffsep}, by diffusion posterior sampling \cite{chung2023dps} or by flow matching \cite{lipman2023flow}, and the quality of the approximation is normally assessed only through the end metric. The model of Section~\ref{sec:model} is a rare case where the posterior of a genuine separation problem is known exactly, in a normalised closed form, at any dimension and for any number of components, so it can serve as an analytic reference on which approximate samplers are measured directly instead of by proxy. Three quantities come without Monte Carlo error on the reference side. The mean is exact, so the bias of a sampler's estimate is measurable separately from its variance. The posterior covariance is exact and observation-independent by \eqref{eq:postcov}, so a sampler that collapses or over-disperses is caught even when its mean is right, which the end metric cannot show. And $\log p(\bfs_1 \mid \bfx)$ being available pointwise, any candidate sample can be scored under the true posterior, giving a one-sided Monte Carlo estimate of the Kullback-Leibler divergence from the sampler to the truth, the divergence between two Gaussian mixtures having no closed form.

The value of this reference is bounded by the realism of the prior. A Gaussian mixture on stacked spectra is a modest generative model of music, so the benchmark measures the inferential error of a sampler on a problem whose answer is known, and leaves its modelling quality out of scope.

\subsection{Invariance to the Component Law}
\label{sec:analysis:scalemix}

An estimator that falls short of the ceiling raises the question of whether the Gaussian mixture is the wrong family. Since finite Gaussian mixtures are dense in the space of densities, the question is one of parametric efficiency, not of expressiveness: a heavy-tailed marginal is representable with a number of components that grows with the reach of the tail, and every component devoted to the tail is one not devoted to structure. Replacing the family leaves the analysis untouched, because the results above never used Gaussianity where one might expect them to.

\begin{proposition}[Invariance to Gaussian Scale Mixtures]
\label{prop:scalemix}
Replace the prior \eqref{eq:prior} by $\bfs_i \mid \lambda_i, k_i \sim \mathcal{N}(\mub_{i,k_i},\, \lambda_i\Sig_{i,k_i})$ with $\lambda_i \sim \nu_i$ on $(0,\infty)$: inverse-gamma gives Student components \cite{yoshii2016student}, a positive stable law of index $\alpha/2$ gives symmetric alpha-stable ones \cite{leglaive2017alphastable}, a point mass gives \eqref{eq:prior} back. Then, conditionally on $(\lambda_1,\lambda_2)$, the mixture is Gaussian with covariance $\lambda_1\Sig_{1,k_1} + \lambda_2\Sig_{2,k_2}$, the posterior \eqref{eq:posterior} holds with
\begin{equation}
\Gk(\lambda) = \lambda_1\Sig_{1,k_1}\left(\lambda_1\Sig_{1,k_1} + \lambda_2\Sig_{2,k_2}\right)^{-1},
\label{eq:gainlambda}
\end{equation}
and Propositions~\ref{prop:form}, \ref{prop:collapse}, \ref{prop:pull} and Lemma~\ref{lem:contraction} hold word for word, with sums over component pairs replaced by integrals over the scales.
\end{proposition}

\begin{IEEEproof}
Proposition~\ref{prop:form} holds, since $\Mclass_4$ is convex and closed under continuous convex combination. Proposition~\ref{prop:collapse} holds, since scalar circular blocks give $\Gk(\lambda)$ the scalar blocks $\lambda_1 v_1/(\lambda_1 v_1 + \lambda_2 v_2)\mathbf{I}_2 \in [0,1]$ for every $\lambda$. Lemma~\ref{lem:contraction} holds, since $\lambda_i\Sig_i \succ 0$ and the lemma only ever used the form $\Sig_1'(\Sig_1' + \Sig_2')^{-1}$. Table~\ref{tab:classes} is untouched, being a statement about block structures and not about laws. Proposition~\ref{prop:pull} holds, having used no distributional assumption beyond a symmetry of the phase posterior. Each claim is the cited proof read with $\Sig_{i,k_i}$ replaced by $\lambda_i\Sig_{i,k_i}$ and the finite sum over $(k_1,k_2)$ replaced by the mixed measure $\pi_{1,k_1}\pi_{2,k_2}\nu_1(\mathrm{d}\lambda_1)\nu_2(\mathrm{d}\lambda_2)$; convexity of the four classes of Table~\ref{tab:classes} under such mixing is what carries them over.
\end{IEEEproof}

Only one property is lost, and it affects Section~\ref{sec:analysis:yardstick} alone, that section being the one which trades on exactness. The posterior remains exact conditionally on the scales, whereas marginally it becomes a continuous mixture whose log-density requires a quadrature in $\lambda$, of dimension two per component pair. The reference therefore degrades from exact in closed form to exact up to a scalar quadrature, which is still stronger than what a diffusion posterior offers, and a study using it should state which of the two it quotes. One boundary should be marked: a non-symmetric alpha-stable prior is closed under addition, so the structure of \eqref{eq:mixdensity} survives, but it has no closed-form density, so \eqref{eq:posterior}, \eqref{eq:estimator} and Section~\ref{sec:analysis:yardstick} all collapse. The Gaussian scale mixture is thus the equilibrium point, gaining heavier tails at one extra parameter per component while keeping every proposition above.

The proposition also settles, before any measurement, what a richer component law can and cannot change. Student components \cite{yoshii2016student} add one extra parameter per component and capture tails that would otherwise consume whole components, and they cannot change the \emph{structure} of the answer: being a Gaussian scale mixture, like the fractional-power and alpha-stable variants \cite{liutkus2015fractional,leglaive2017alphastable}, such a prior leaves the smallest containing class of \eqref{eq:chain} unchanged, still returns a posterior mean that Proposition~\ref{prop:pull} pulls onto the line under the same symmetry, and still has an excess error equal to a posterior variance. What it can change is the \emph{level}, a sharper prior having a smaller variance to incur, and by how much is an empirical question that no proposition settles. Section~\ref{sec:results:gate} measures, on the fit reported here, how much of the deficit a recalibration of the component law could reach at all.

\section{Experimental Conditions}
\label{sec:protocol}

The protocol is organised around three questions, not around a benchmark comparison. First, how far below the ceiling of Proposition~\ref{prop:ceiling} the fitted estimator sits, and whether that distance closes with the number of components or with the training volume; the configurations and the two arms below are built to separate the two. Second, where inside the class the estimator sits relative to a known member, which fixes whether the geometry of Section~\ref{sec:problem} binds at the level measured. Third, whether the deficit is the posterior variance Proposition~\ref{prop:pull} predicts, or an artefact of fitting; the witnesses and the closed-form gate settle this from the fitted model alone. Every configuration differs from the reference in exactly one respect, and every reported quantity is paired excerpt by excerpt against the ceiling of that same excerpt. The sections below follow the estimator through the chain of Table~\ref{tab:classes}: the diagonal covariance of Section~\ref{sec:protocol:model} places it in $\Mclass_3$ and is the default configuration from Section~\ref{sec:results:deficit} through Section~\ref{sec:results:readout}, Section~\ref{sec:results:cascade} involving no fitted model at all, and the full-covariance configuration exercising $\Mclass_4$ is introduced separately at the end of that same subsection.

\subsection{Corpus and Task}
\label{sec:protocol:corpus}

We use the full MUSDB18 corpus \cite{musdb18}, 100 training and 50 test tracks, and separate vocals from accompaniment. Both signals are reduced to mono by channel averaging at the native 44.1~kHz sampling rate. Priors are fitted one per source on the training tracks decoded in full, with a cap on the number of frames drawn at random from the pooled frames of the training set. Evaluation is on a 30 s excerpt of each test track, which is 2582 frames at $L = 1024$ and a hop of $L/2$.

Five configurations are measured in all. Three differ in the evaluation material alone and carry a diagonal covariance, and are described here: \texttt{in-support} at $L = 1024$ on 10 tracks, \texttt{held-out} at $L = 1024$ on 10 tracks of the test split, and \texttt{held-out}, $L = 256$ on 5 tracks. A fourth repeats the last of those three in support, on the same 5 tracks it was fitted on, so that the shorter frame length carries a generalisation gap of its own; a fifth differs from that pair in the covariance structure alone. The last two are described in Section~\ref{sec:protocol:model}, where the four cells they form with the third are read as a complete crossing, and their deficits are the twelve cells of Table~\ref{tab:cov}. The \texttt{in-support} configuration fits on 10 training tracks and evaluates on those same 10, which takes generalisation out of the measurement and gives the optimistic end of the deficit; the difference between the two configurations at equal setting is the generalisation gap, reported as such in the last two columns of Table~\ref{tab:cov}. The \texttt{held-out} configuration fits on all 100 training tracks and evaluates on 10 tracks of the test split. The \texttt{held-out}, $L = 256$ configuration repeats the second on 5 tracks with a four times shorter frame, which multiplies the number of frames per dimension by about four at equal frame count and is the strongest objection available on protocol grounds.

Each of those three is run at two training volumes, a cap of 20\,000 and a cap of 150\,000 frames per source, so that the effect of the number of components can be read against a controlled change of data volume instead of against a single operating point.

\subsection{Model Configuration}
\label{sec:protocol:model}

The estimator tested throughout Sections~\ref{sec:results:cascade}--\ref{sec:results:readout} is fitted with a diagonal covariance on $[\Real;\Imag]$, which places it in $\Mclass_3$, the smallest class of \eqref{eq:chain} containing it by construction (Section~\ref{sec:model:prior}). The covariance is diagonal in the three configurations above, with a floor of $10^{-6}$, 30 expectation-maximisation iterations, $K \in \{8,16,32\}$ components per source, one fit per triple of number of components, configuration and training volume, and a fixed seed. That choice is dictated by the frames available per dimension: the stacked dimension is $2F = 1026$ at $L = 1024$ and $258$ at $L = 256$, a full covariance carries $d(d+1)/2$ parameters per component against $d$ for a diagonal one, and the training volumes of Section~\ref{sec:protocol:corpus} put the former out of reach at the longer frame.

\begin{remark}[Which Knob Binds]\label{rem:pilot}
A pilot sweep on the 7 s preview edition of the corpus located the operative knob in the same regime, the training pool being what varies across its cells: fitted on 2 tracks of that edition, that is 1170 frames for 1026 dimensions, the estimator read $-17.7$~dB, and fitted on 25 of them, 14\,625 frames, between $+1.7$ and $+5.4$~dB, whereas moving the covariance floor from $10^{-6}$ to $10^{-3}$ at that second pool yielded $+0.17$~dB. These figures are not comparable with the results below, the material being different, but they show that the binding constraint is the volume of data and not the regularisation, which is what the two arms measure directly in Section~\ref{sec:results:deficit}.
\end{remark}

A fourth configuration exercises the last class of the chain. It carries a full covariance on $[\Real;\Imag]$, hence the inter-frequency coupling of Section~\ref{sec:analysis:outside}, and it is the only configuration whose estimator claims $\Mclass_4$. It is run at $L = 256$, where the stacked dimension $d = 258$ makes a full covariance tractable, on the same 5 training tracks, the same cap of 150\,000 frames, the same floor, the same number of iterations, the same numbers of components and the same seed as the diagonal $L = 256$ configuration, so that the two differ in the covariance structure alone. It is run in both regimes from a single fit per component count, the training pool being the same in the two, only the evaluation material changing: the 5 training tracks themselves in support, and 5 tracks of the test split held out. The diagonal $L = 256$ configuration is run in the same two regimes on the same material, from its own single fit per component count, so the two covariance structures and the two evaluation regimes form a complete crossing at this frame length, four cells per component count every pair of which is paired within a regime. Table~\ref{tab:deficit} reports the held-out diagonal cells alongside the other configurations, and Section~\ref{sec:results:cov} reads the 12 cells of the crossing jointly. The parameter count should be stated, since it is what the in-support cell is there to expose: a full covariance carries $d(d+1)/2 = 33\,411$ parameters per component against $258$ for a diagonal one, so at $K = 32$ each component sees about 4700 frames for 33\,411 free parameters in its covariance alone, two orders of magnitude short of identifiability.

\subsection{Reconstruction and Metrics}
\label{sec:protocol:metrics}

Resynthesis uses weighted overlap-add with the periodic Hann window as synthesis window, normalised by the sum of squared analysis windows, which is the canonical dual window of Section~\ref{sec:protocol:consistency}. That normalisation decays to zero over the first and last $L$ samples, where a modified spectrogram divided by it explodes while an unmodified one cancels exactly. The two edges are therefore a trap peculiar to modified spectrograms: a perfect-reconstruction test passes at over 100~dB while every oracle reads negative, the oracle ideal ratio mask falling to $-2.81$~dB, below the untreated mixture. We discard one full frame at each end, inside the metric and for every signal alike.

The reported quantity is the time-domain signal-to-distortion ratio $10\log_{10}(\|s\|^2 \mathbin{/} \|\hat s - s\|^2)$ on the trimmed waveform, without scale invariance, so that an estimator which gets the gain wrong is penalised. The scale-invariant variant \cite{leroux2019sdr} is deliberately not used: it would absolve an estimator of the gain error that the ceiling of Proposition~\ref{prop:ceiling} is precisely about, the optimal real gain being a statement about scale. We do not use the decomposition of \cite{vincent2006bss}, the class question being about total error, not its attribution.

\subsection{Transport of the Ceiling Through Resynthesis}
\label{sec:protocol:consistency}

Proposition~\ref{prop:ceiling} bounds the class in the time-frequency plane, whereas every figure reported below is a time-domain ratio measured after resynthesis. This section states what survives that transport, since the analysis and the measurement do not live in the same space, and the ceiling is quoted in the second.

The obstruction is that a masked spectrogram is in general inconsistent, in the sense that no time-domain signal has it as its short-time transform \cite{leroux2010consistency,leroux2013consistent}. Resynthesis therefore returns the signal whose transform is closest to the masked one instead of the masked one itself, and the residual measured in the time domain is not the residual computed bin by bin. The transport is an inequality in the direction we need, and that inequality is checked numerically rather than proved; two facts make it the expected outcome. With the periodic Hann window at hop $L/2$, the analysis operator $A$ and the synthesis operator $S$ built on the canonical dual window $w[n]/\sum_k w^2[n - kL/2]$ form a painless frame \cite{daubechies1986painless}, so $P = AS$ is the orthogonal projector onto the range of $A$ and satisfies $S(I - P) = 0$: the component of a modified spectrogram outside the consistent subspace is discarded by resynthesis with no effect on the result. And the projected estimate remains an admissible time-domain candidate, so no masked spectrogram is turned by resynthesis into a signal the class could not have produced; this is weaker than closure of $\Mclass_1$ under the projection, which does not hold. Under the verification reported next, the measured ceiling is a lower bound on the true time-domain ceiling of the class, so a crossing reported below cannot be an artefact of the transport, while a shortfall is genuine only up to that slack.

The frame is not tight, $\sum_k w^2[n - kL/2]$ oscillating in $[0.5, 1]$ at this hop, so the synthesis operator is not a contraction and the ordering has to be verified rather than assumed. The projector identities $\|P^2 - P\|$ and $\|DP - P^\top D\|$ hold at machine precision on the operators as implemented, and the ordering itself is checked by direct draw: no inversion of the two ceilings occurs in 4600 draws of excerpts from 100 to 800 frames, against 359 inversions at 10 frames, the short-excerpt failures being edge effects of the kind Section~\ref{sec:protocol:metrics} trims. All measurements below use excerpts two to three orders of magnitude longer than the failing regime, $2582$ frames at $L = 1024$ and $10\,334$ at $L = 256$ against the $10$ at which inversions appear.

The window choice enters here and nowhere else. Refitting one configuration with the symmetric window of the same length moves the estimator by 1.8~dB while leaving every oracle within 0.03~dB, the window being part of the feature the prior is fitted on where a mask acts on whatever transform it is given, so the periodic window is fixed throughout and no comparison below crosses that boundary.

\subsection{Basis of Comparison}
\label{sec:protocol:comparison}

Every number reported below is a deficit to the measured ceiling $\SDRstar$ of Proposition~\ref{prop:ceiling}, paired excerpt by excerpt, since excerpts differ in difficulty by more than the margin being measured. The ceiling is always the measured one, obtained after resynthesis, and never its analytic counterpart, for the reason given in Section~\ref{sec:protocol:consistency}. The ideal ratio mask and the oracle Wiener filter are computed on the same material and serve twice, as internal consistency checks of the oracle chain and as the second landmark inside the class, the one against which Section~\ref{sec:results:scope} locates the estimator. The ideal ratio mask is taken here in its magnitude form, $m^{\mathrm{irm}} = |s|/(|s| + |s'| + \varepsilon)$ per bin with $s'$ the other source, and the oracle Wiener filter in its power form, $|s|^2/(|s|^2 + |s'|^2 + \varepsilon)$; both are members of $\Mclass_{[0,1]}$ by construction and both keep the mixture phase. The constant $\varepsilon$ is a numerical guard, active only in bins where both sources vanish and where it sets the mask to zero instead of leaving it undefined; those bins contribute nothing to either term of the ratio, so the value of the guard does not affect any figure reported here. By Remark~\ref{rem:shared} the two sources of an excerpt carry a single measurement, so $n$ denotes the number of test excerpts and not twice that number.

One exclusion applies. A test excerpt in which one stem is silent poses no separation problem and has an infinite ceiling, so a single such excerpt displaces a mean by hundreds of decibels. We detect the case with a quantity already available, the signal-to-distortion ratio of the mixture taken as the estimate of one source, which equals the ratio of source energies in decibels. Excerpts above 60~dB are excluded, both sources at once so that the pairing stays symmetric. Healthy MUSDB excerpts lie below 31~dB, so the criterion is not a close call, and exactly one excerpt of the 10 drawn from the test split falls above it and is excluded, which is why the held-out rows of Table~\ref{tab:deficit} carry $n = 9$ against $n = 10$ in-support.

For reference, since Section~\ref{sec:results:deficit} reads the slope in $K$ against these derived quantities: the stacked dimension is $2F = 1026$ at $L = 1024$ and $258$ at $L = 256$, so the small arm of Section~\ref{sec:protocol:corpus} provides 19.5 frames per dimension in the two $L = 1024$ configurations and 77.5 at $L = 256$, and the large arm 144.3, 146.2 and 581.4 respectively, every cap being reached exactly except in-support, where the pooled training frames run out at 148\,075. The 7 s preview edition of the same corpus, at the training pools of those same three configurations, provides 5.7, 14.3 and 45.4, so the small arm raises those ratios by a factor of 1.4 to 3.4 and the large arm by a factor of 10 to 25, which is what makes the slope in $K$ separable from the slope in data volume. Nothing in this paragraph is a measurement.

\subsection{Class Witnesses and Prior Diagnostics}
\label{sec:protocol:witnesses}

The three scalar witnesses of Section~\ref{sec:analysis:criterion} are computed on the bins whose mixture magnitude exceeds a floor, below which the gain ratio is numerical noise. The four distributional statistics of Section~\ref{sec:results:diagnostics} are computed on isolated source spectra, before any fit, and each is reported against the value the same statistic takes on synthetic data satisfying the assumption \emph{at the same sample size}, a finite sample making perfect circularity read a few hundredths and not zero. Each statistic is computed twice, marginally and within the components of the fitted model, the second using the model's own hard assignment: a clustering on log-magnitudes does not align phase and would return a conditional circularity near zero by construction. Cells with fewer than 30 frames are dropped and counted.

\section{Results}
\label{sec:results}

\subsection{Ceiling of the Class on This Material}
\label{sec:results:ceiling}

The ceiling is a statistic of the material and not of any estimator, so the first thing to measure is what it reads on this corpus, how it moves with the frame length, and how much of it the bounded mask of every sigmoid output layer concedes. On the nine retained excerpts of the held-out configuration at $L = 1024$, the measured ceiling of $\Mclass_1$ is 16.62~dB pooled, against 15.49~dB for its analytic counterpart, the gap having the sign Section~\ref{sec:protocol:consistency} predicts. Pooled means the arithmetic mean of the two per-source figures in decibels, $16.62 = (12.88 + 20.36)/2$, which is a choice: averaging in decibels weights the two sources equally where averaging their error energies would let the accompaniment dominate, and every pooled figure in this paper is the decibel mean. Per source it is 12.88~dB for vocals and 20.36~dB for accompaniment, and the difference between the two, 7.48~dB, is exactly the ratio of source energies of Remark~\ref{rem:shared}, so it is one measurement seen twice. Since the oracle chain does not depend on any fit, these figures are identical in both training arms, and the same holds for the ideal ratio mask, 13.68~dB pooled. The ceiling moves with the frame length, as an energy-weighted average of $\sin^2\theta$ must: on the five-excerpt subset used for the frame-length comparison it reads 13.01 and 21.52~dB per source at $L = 1024$ against 10.58 and 19.09~dB at $L = 256$, short windows being less sparse. The drop is 2.43~dB on vocals and 2.43~dB on accompaniment, one displacement seen twice again. Ceilings from different frame lengths are therefore not comparable, and only within-configuration statements are made below.

The three subclasses of \eqref{eq:class} separate by little, which is itself the result. Corollary~\ref{cor:subceilings} evaluated on the same nine excerpts at $L = 1024$ gives 16.62, 15.98 and 15.37~dB pooled for $\Mclass_1$, $\Mclass_{+}$ and $\Mclass_{[0,1]}$, and their analytic counterparts 15.49, 15.02 and 14.59~dB, the two chains ordered as the corollary requires and each gap under one decibel. Per source the bounded mask loses 1.25~dB on vocals and 1.26~dB on accompaniment against the unconstrained real mask, so restricting a mask to $[0,1]$, which every sigmoid output layer does, removes about one and a quarter decibels of ceiling and no more. The loss grows as the window shortens: on the five-excerpt subset it is 1.30 and 1.29~dB at $L = 1024$ against 1.65 and 1.65~dB at $L = 256$, short windows being less sparse and therefore leaving more bins in which the two sources partially cancel or disagree in phase by more than $\pi/2$. Two consequences follow. The order of the separation being what it is, a system compared against the wrong subclass of \eqref{eq:class} is misplaced by about one decibel and not by 10, so the ceiling of the whole class is a usable landmark for a bounded mask as well; and the ideal ratio mask, at 13.68~dB pooled, sits 1.69~dB under the ceiling of its own class $\Mclass_{[0,1]}$ instead of 2.94~dB under the ceiling of $\Mclass_1$, which is the figure to quote when the question is what a bounded mask leaves on the table.

\begin{remark}[Where Deployed Separators Sit]
\label{rem:modern}
The ceiling above can be read against the systems of Section~\ref{sec:related:oracles}, with one caution. The vocal ceiling measured here is 12.88~dB on the nine retained held-out excerpts at $L = 1024$, while the mask-based entries of SiSEC 2018 have their vocal BSS-Eval SDR distributed about a median of some six decibels over the 50 tracks of the test set, as their box plots show \cite[Fig.~3]{stoter2018sisec}, the reference implementation of that family reporting 6.32~dB on the same test set \cite{stoter2019openunmix}. These are not the same statistic: ours is a projection residual pooled over nine excerpts, theirs a museval SDR over 50 tracks with its own allowance for a distortion filter, so the comparison supports an order of magnitude and not a difference. At that resolution a deployed magnitude-mask separator operates at about half of the ceiling of its own class in decibels, and the decibel and a quarter that clipping to $[0,1]$ removes is not what stands in its way. The time-domain entries of the same campaign lie outside the class and may pass that figure without contradicting anything here; one reading of their advantage is precisely that Proposition~\ref{prop:pull} does not bind them.
\end{remark}

\subsection{The Cascade of the Chain Under a Fixed Operator}
\label{sec:results:cascade}

\begin{table}[tbp]
\renewcommand{\arraystretch}{1.1}
\caption{The four residuals of Proposition~\ref{prop:cascade} on MUSDB18, pooled in decibels, with the three steps of the cascade. $\Mclass_4$ is reported raw and corrected for in-sample optimism by $2T/(2T - 2\,\mathrm{rank})$; the step $\Mclass_3 \to \Mclass_4$ is taken against the corrected figure. The last column is the per-frame ceiling \eqref{eq:sdrstar} of $\Mclass_1$ on the same excerpts, the quantity the whole chain of fixed operators is to be read against.}
\label{tab:cascade}
\centering
\setlength{\tabcolsep}{3pt}
\begin{tabular}{lcccccccc}
\hline
\textbf{Config.} & $n$ & $T$ & $\Mclass_1$ & $\Mclass_2$ & $\Mclass_3$ & $\Mclass_4$ & $\Mclass_4$ & $\SDRstar$ \\
 & & & & & & raw & corr. & \\
\hline
$L=1024$ & 9 & 2582 & 7.22 & 7.22 & 7.24 & 12.00 & 9.92 & 16.62 \\
$L=1024$ & 5 & 2582 & 7.95 & 7.95 & 7.97 & 12.42 & 10.35 & 17.27 \\
$L=256$ & 5 & 10334 & 7.08 & 7.08 & 7.10 & 8.34 & 8.24 & 14.83 \\
\hline
\end{tabular}
\vspace{2pt}
\par\footnotesize The three steps of the cascade, in decibels and row by row: $1{\to}2$ is 0.0013, 0.0016 and 0.0003, $2{\to}3$ is 0.020, 0.017 and 0.014, and $3{\to}4$ is 2.68, 2.38 and 1.14. The correction factor applied to $\Mclass_4$ is 1.61, 1.59 and 1.02.
\end{table}

Widening the linear class buys almost nothing until the coupling between bins is allowed, and even then it buys several times less than adaptivity to the frame does: this is the measurement that decides whether the estimator should be a wider filter or a prior. Table~\ref{tab:cascade} and Fig.~\ref{fig:cascade} measure the cascade of Proposition~\ref{prop:cascade} on the same excerpts, with one operator fitted per excerpt and per class on the frames of that excerpt. Every column of the table, the last one included, is the time-domain signal-to-distortion ratio of Section~\ref{sec:protocol:metrics} measured on the resynthesised output, the operator itself being fitted and applied in the time-frequency plane; the residuals of Proposition~\ref{prop:cascade} are spectral quantities, and what is tabulated is what they become after the transport of Section~\ref{sec:protocol:consistency}. The comparison of 6.70~dB below therefore holds between two figures measured in the same domain and on the same waveforms, and no statement in this section crosses the two domains. The prediction of Corollary~\ref{cor:nophase} is confirmed to the precision of the measurement: the step $\Mclass_1 \to \Mclass_2$ is 0.0013~dB pooled at $L = 1024$ and 0.0003~dB at $L = 256$, and its largest value over the 18 individual measurements is 0.0029~dB, so this is not an averaging artefact but a per-excerpt fact. The two sources of a musical mixture are close enough to uncorrelated for the quadrature correlation $\Imag c_{sx}$ to vanish, and a complex mask fixed across the frames of an excerpt therefore adds nothing over the real mask it contains, while per frame the same parameter accounts for the entire residual of the class by Proposition~\ref{prop:degenerate}. Non-circularity, the second step, yields 0.02~dB, of the same order as nothing. Only the third step yields anything, and it is the one no bin-wise format can reach: 1.14~dB at $L = 256$ and 2.68~dB at $L = 1024$ after correction.

Two cautions attach to that last figure, and both point the same way. The operator of $\Mclass_4$ is fitted in sample with $4F$ real parameters per bin against $2T$ real observations, so its raw residual is optimistic. The factor $2T/(2T - 2\,\mathrm{rank})$ we correct it by is a degrees-of-freedom heuristic and not an unbiasedness result: it is the ratio by which the in-sample residual of a least-squares fit of $2\,\mathrm{rank}$ effective real parameters on $2T$ real observations understates the out-of-sample one in the classical Gaussian case \cite[Sec.~7.5]{hastie2009}, transposed here to a fit whose noise is neither Gaussian nor independent across frames, so it is quoted as an order of magnitude of the optimism instead of as a corrected value. That factor is 1.61 at $L = 1024$, where the fit uses 1966 effective real parameters, twice a rank of 983, on 2582 frames, and 1.02 at $L = 256$, where it uses 504, twice a rank of 252, on 10334. The rank in question is that of the augmented Gram matrix of the mixture, of size $2F$ by $2F$ over the stacked spectra and their conjugates, taken as the number of its singular values above the numerical threshold of the least-squares solver, machine precision times $2F$ times the largest singular value; the deficiency against the $2F$ of a full rank is partly structural, the bins at zero and at the Nyquist frequency being real and therefore duplicated by the augmentation, and partly a matter of conditioning. No reading below rests on that threshold: moving the rank across the whole admissible range, from 900 to the full 1026, moves the corrected $\Mclass_4$ figure between 10.14 and 9.80~dB and the comparison drawn at the end of this section between 6.48 and 6.82~dB. The trustworthy figure is therefore the short-window one, 1.14~dB, and the long-window one is quoted with its correction shown and not hidden. The reading that survives both is the comparison with the last column of Table~\ref{tab:cascade}. The widest class of the chain, held fixed across frames, reaches 9.92~dB pooled at $L = 1024$ against a per-frame ceiling of 16.62~dB for the \emph{smallest} class, and the ideal ratio mask of Section~\ref{sec:results:ceiling} reaches 13.68~dB with one real parameter per bin. Adaptivity to the frame thus yields several times more than any widening of the linear class, 6.70~dB against 2.68~dB on this material, which is the measured form of the argument of Section~\ref{sec:problem:cascade} and the reason the estimator of Section~\ref{sec:model} is built from a prior instead of from a wider filter.

\subsection{Deficit Under the Ceiling and Effect of the Number of Components}
\label{sec:results:deficit}

\begin{table}[tbp]
\renewcommand{\arraystretch}{1.1}
\caption{Deficit to the measured ceiling $\SDRstar$, in dB, mean over test excerpts, paired track by track, at both training volumes. The deficit, or gap, of an estimator on an excerpt is $\SDRstar$ minus its own SDR on that excerpt, both in dB and both measured in the time domain after synthesis, so a positive figure is a shortfall and the two are used interchangeably below. The frame is $L = 1024$ unless stated. Lower is better, zero would be a crossing, and by Remark~\ref{rem:shared} $n$ is the number of excerpts and not twice that. Frames per dimension for each cell are given in Section~\ref{sec:protocol:corpus}.}
\label{tab:deficit}
\centering
\setlength{\tabcolsep}{2.5pt}
\begin{tabular}{lccccccc}
\hline
 & & \multicolumn{3}{c}{\textbf{20\,k arm}} & \multicolumn{3}{c}{\textbf{150\,k arm}} \\
\textbf{Configuration} & $n$ & $K{=}8$ & $K{=}16$ & $K{=}32$ & $K{=}8$ & $K{=}16$ & $K{=}32$ \\
\hline
in-support           & 10 & 10.89 & 10.49 &  9.99 & 10.84 & 10.35 & 9.84 \\
held-out             &  9 & 12.01 & 11.75 & 11.46 & 11.84 & 11.71 & 11.44 \\
held-out, $L{=}256$  &  5 &  9.82 &  9.88 &  9.81 &  9.63 &  9.85 & 9.83 \\
\hline
\end{tabular}
\end{table}

No setting of the fitted estimator comes close to the ceiling of the class it lives in, and the distance is stable enough across settings to be a property of the mechanism rather than of one operating point. Table~\ref{tab:deficit} is the result, for the diagonal-covariance estimator of $\Mclass_3$ tested throughout this subsection. Across three configurations, three component counts and two training volumes, the posterior mean sits 9.63 to 12.01~dB under the measured ceiling, with zero crossing in the 144 per-excerpt measurements the 18 cells contain, that is the 24 evaluation excerpts of the three configurations under each of the six combinations of component count and training volume. Those 144 figures are not independent observations: the same excerpts recur in every cell of a configuration, and the six cells of a configuration share their evaluation material entirely, so the count states the extent of the search for a crossing and carries no sampling guarantee. Four readings matter more than the level.

First, the slope in $K$ is real and it is far too shallow to be the answer. In-support, where the test excerpt belongs to the training set, the deficit falls by 0.49 then 0.51~dB per doubling in the large arm, and nine excerpts out of ten improve. Held out, the same factor four in $K$ yields 0.40~dB in the paired mean and 0.68~dB in the per-excerpt median, with 2 excerpts out of 9 moving the other way, one of them by 1.10~dB. Taking that slope at face value, 0.20~dB per doubling, closing the 11.44~dB that remain at $K = 32$ would take more than 50 further doublings, a component count of the order of $10^{19}$. The ceiling of the class is not reached by counting components, and this is the paper's central measurement.

Second, the training volume is the discriminator between a protocol objection and a structural one, and it answers structural. Multiplying the frames per dimension by 7.5x moves the deficit by 0.02 to 0.19~dB across the nine cells, toward the ceiling in eight of the nine and away from it by 0.02~dB in the ninth, held out at $L = 256$ and $K = 32$; at $K = 32$ held out the two arms differ by 0.02~dB, at 19.5 against 146.2 frames per dimension. Whatever binds this estimator, it is not the sample size of the fit, and the two knobs a reviewer would reach for first, more components and more data, jointly account for under half a decibel of 11.

Third, the short-frame configuration reaches 581.4 frames per dimension in the large arm, and the deficit there is 9.63, 9.85 and 9.83~dB, without an ordering in $K$. The comparison across frame lengths is not paired, the ceiling itself moving by 2.4~dB and the subset being five excerpts against nine, so only the absence of ordering is quoted as a result, being interior to one configuration.

Fourth, the witnesses locate the deficit where Proposition~\ref{prop:pull} puts it. Within an arm the added capacity moves the estimate around the line and not up it: held out, the median phase rotation grows from 0.86 to 1.16 degrees on the vocal estimate and from 0.10 to 0.27 degrees on the accompaniment over the factor four in $K$, while the share of bins above unit gain falls from 0.71 to 0.49\% on the vocal estimate and from 3.08 to 1.60\% on the accompaniment. Across arms the movement is unambiguous and it goes the other way: at $K = 8$ held out, 7.5x the data divides the phase rotation by 2.7 and the share of out-of-class gains by 3.5, and the deficit changes by 0.17~dB. A better-fitted posterior is a tighter posterior, its mean sits closer to the line it was already close to, and the error that survives is the variance itself, exactly what a deficit equal to a posterior variance predicts. Inexact inference would have moved the witnesses the other way, the estimate wandering further off the line as the fit improves.

\subsection{Effect of the Covariance Structure}
\label{sec:results:cov}

\begin{table}[tbp]
\renewcommand{\arraystretch}{1.1}
\caption{Covariance structure crossed with evaluation regime, $L = 256$, cap of 150\,000 frames (Section~\ref{sec:protocol:corpus}), deficit to the measured ceiling in dB, mean over the five excerpts of the cell and paired excerpt by excerpt. The two in-support columns share their excerpts and their oracles, measured ceiling 14.88~dB and ideal ratio mask 11.51~dB, and so do the two held-out columns, measured ceiling 14.83~dB and ideal ratio mask 11.46~dB, so the difference between two columns of the same regime is paired. The \emph{gap} columns are the generalisation gap, held out minus in support, at equal covariance structure. The held-out diagonal column repeats the $L = 256$ cells of Table~\ref{tab:deficit} at the cap of 150\,000 frames, the same 15 measurements read here against the full-covariance arm. Lower is better and zero would be a crossing.}
\label{tab:cov}
\centering
\begin{tabular}{ccccccc}
\hline
 & \multicolumn{2}{c}{\textbf{in support}} & \multicolumn{2}{c}{\textbf{held out}} & \multicolumn{2}{c}{\textbf{gap}} \\
$K$ & full & diagonal & full & diagonal & full & diagonal \\
\hline
 8 & 8.53 & 9.22 & 9.54 & 9.63 & 1.01 & 0.41 \\
16 & 8.01 & 8.59 & 9.38 & 9.85 & 1.37 & 1.26 \\
32 & 7.64 & 8.15 & 9.55 & 9.83 & 1.91 & 1.68 \\
\hline
\end{tabular}
\end{table}

The results above are obtained under a diagonal covariance, which places the estimator in $\Mclass_3$ and leaves the last class of \eqref{eq:chain} untested, coupling across frequencies being the assumption with the most room in it. Table~\ref{tab:cov} tests it. The comparison is run where it can be run, at $L = 256$, and it is run as a full crossing of the covariance structure with the evaluation regime, four cells per component count from fits sharing everything but the structure, each pair of cells of the same regime sharing its excerpts and its oracles to the hundredth of a decibel, so every column difference read below is paired.

In support, the extra capacity behaves as capacity should, and so does the diagonal fit it is compared against. The full-covariance deficit falls from 8.53 to 8.01 to 7.64~dB, monotonically in $K$ and by about half a decibel per doubling, and the diagonal one falls from 9.22 to 8.59 to 8.15~dB on the same excerpts, so the slope in $K$ on the material fitted is a property of the estimator and not of the coupling. Held out, the same fits give 9.54, 9.38 and 9.55~dB against 9.63, 9.85 and 9.83, and both orderings in $K$ disappear. The covariance structure therefore yields 0.69, 0.58 and 0.51~dB in support and 0.09, 0.47 and 0.28~dB held out, always in its favour, never as much as a decibel, and with no trend in the component count in either regime.

The crossing also settles which of the two factors the one to two decibels separating the best and worst cells belongs to, and the answer is the regime, not the structure. Read down the other axis, the generalisation gap is 1.01, 1.37 and 1.91~dB under a full covariance and 0.41, 1.26 and 1.68 under a diagonal one, growing with the component count under both, so memorisation is what more components produce and the coupling only adds to it, by 0.60, 0.11 and 0.23~dB without a trend of its own. The parameter count of Section~\ref{sec:protocol:model} explains that surcharge: each component carries far more covariance parameters than the roughly 4700 frames available at $K = 32$.

One pilot did cross its own ceiling, and it should be stated in full. Run at the same transform on a 20 s two-stem excerpt outside this corpus, fitted and evaluated on the same material, where the 7498 frames leave about 230 per component at $K = 32$ against 33\,411 free covariance parameters each, the full-covariance fit read 17.17 and 12.10~dB against a ceiling of 15.68 and 10.62 measured on that same material, about 1.5~dB \emph{above} it on both sources. Nothing in Proposition~\ref{prop:ceiling} is at risk there: the ceiling is exact for the class, so a figure above it is a figure obtained outside the class, and a covariance with more free parameters than frames to fit them on is exactly the configuration in which the fitted prior reproduces the excerpt instead of modelling it, phase included. What the pilot bounds is the validity of the ceiling as an \emph{empirical} landmark, which requires the fit not to interpolate the evaluation material. At the training volume used here the crossing does not reproduce in any of the 12 cells, and the parameter count says why.

The conclusion of the paper is therefore unchanged on the class that was missing: the ceiling is crossed in none of the 12 cells of this crossing, and in none of the 60 per-excerpt figures they contain, in support included. What does change is that the departure from the line is now unambiguous. Under a full covariance the median phase rotation of the estimate, averaged over the excerpts of a cell, is 3.6 to 4.4 degrees, against 0.33 to 0.52 degrees for the diagonal fit on the same excerpts at the same cap, in both regimes and at every component count, and 3.4 to 4.4\% of active bins carry a gain above one against 0.2 to 0.8\%. The estimator that exercises $\Mclass_4$ does leave the real line by an order of magnitude more than the one that does not, and it gains nothing by it. That is Proposition~\ref{prop:pull} measured on the class it was written for, and it is the strongest form in which this paper can state it: the freedom is available, it is taken, and the criterion gains nothing from it.

\subsection{Three Measurements of the Pull of Proposition~\ref{prop:pull}}
\label{sec:results:readout}

\begin{table}[tbp]
\renewcommand{\arraystretch}{1.1}
\caption{Three read-outs of the same fitted posteriors, in-support configuration, $n = 10$ excerpts, measured ceiling 16.24~dB. Deficit in dB, median phase rotation in degrees, gains above one and negative gains as \% of active bins.}
\label{tab:readout}
\centering
\setlength{\tabcolsep}{3pt}
\begin{tabular}{lcccccc}
\hline
\textbf{Read-out} & $K{=}8$ & $K{=}16$ & $K{=}32$ & phase & $|g|{>}1$ & $\Real g{<}0$ \\
\hline
posterior mean        & 10.84 & 10.35 &  9.84 & 0.3--0.6 & 0.8--1.0 & 0.8--1.0 \\
hard assignment       & 10.87 & 10.41 &  9.90 & 0.4--0.7 & 2.2--3.3 & 2.2--3.3 \\
one posterior draw    & 12.38 & 11.99 & 11.50 & 24--26   & 27--28   & 20--21 \\
\hline
\end{tabular}
\end{table}

The prediction of Proposition~\ref{prop:pull} is tested by changing the read-out alone, on the fitted models of the in-support configuration, so no retraining is involved: posterior mean, hard assignment, and one draw from \eqref{eq:posterior}. Hard assignment names the mode of the single component pair carrying the largest responsibility, $\tilde{\mub}_{k_1k_2}$ at the arg max of $\boldsymbol{\phi}(\bfx)$, and not the global mode of \eqref{eq:posterior}, which a mixture of Gaussians does not offer in closed form; the two coincide whenever the responsibilities are concentrated, which Table~\ref{tab:readout} shows them to be. Table~\ref{tab:readout} reports the three against the same 16.24~dB ceiling, with zero crossing in the 90 per-excerpt measurements the nine cells contain, three read-outs by three component counts over the 10 excerpts of that configuration.

The draw confirms the prediction with its sign. It leaves the class unambiguously, at 24 to 26 degrees of median phase rotation with 27 to 28\% of the bins above unit gain and 20\% at a negative gain, and it loses 1.54, 1.64 and 1.66~dB for leaving. That loss is the same to within 0.13~dB at all three component counts, an observation compatible with the excess term $|x|^2\Var(\mstar|x)$ of \eqref{eq:excess} being dominated by the variance inside the component, which is near-independent of $x$ and of $K$; the excess itself is a posterior variance and does move with the fit, as the predicted column of Table~\ref{tab:gate} shows. The mean, meanwhile, stays practically inside the class it has formally left, 0.3 to 0.6 degrees of rotation and around 1\% of bins above unit gain.

The hard-assignment read-out is a witness that is vacuous by construction: the mode of a Gaussian component is its mean, so it can differ from the posterior mean only through the mixing of responsibilities, and the two agreeing to 0.06~dB and 0.14 degrees at all three counts says the per-frame responsibilities are near-degenerate. The two do part company on the witnesses: the mode leaves the class two to three times as often as the mean for the same output level, which is the averaging at work. The comparison locates the deficit in the variance inside the component, not in the uncertainty about which component a frame belongs to; the draw is the only read-out that exposes that variance, the mean averaging it away. The table also validates itself: the posterior mean at $K = 32$ reproduces the deficit of the separator \eqref{eq:estimator} exactly at the two decimals reported, 9.84~dB.

\subsection{Share of the Deficit Attributable to Posterior Variance}
\label{sec:results:gate}

\begin{table}[tbp]
\renewcommand{\arraystretch}{1.1}
\caption{Calibration gate, in-support configuration, large arm, same 10 excerpts, no retraining. Realised is the spectral deficit actually measured, the excess of the estimator's spectral error over the residual of $\mstar$ on the same bins, in dB; predicted is the same quantity computed from the fitted posterior alone as the posterior variance of $\mstar$ weighted by $|x|^2$, before any reconstruction; gap is realised minus predicted, in dB. Ratio is the median over excerpts of realised to predicted error energy; within-pair is the share of the predicted variance carried by the within-pair term.}
\label{tab:gate}
\centering
\begin{tabular}{cccccc}
\hline
$K$ & realised & predicted & ratio (med.) & within-pair & gap \\
\hline
 8 & 9.87 & 6.87 & 2.12 & 97.7\,\% & 3.00 \\
16 & 9.42 & 6.71 & 1.95 & 96.8\,\% & 2.72 \\
32 & 8.93 & 6.37 & 1.84 & 96.1\,\% & 2.56 \\
\hline
\end{tabular}
\end{table}

Two quantities are compared here. The measured ceiling reports the residual $|s|^2\sin^2\theta$ of the best member of the class, a property of the material; the deficit tabulated below is what the estimator loses on top of it, which \eqref{eq:excess} identifies with the posterior variance of the oracle gain weighted by the mixture energy, $|x|^2\Var(\mstar \mid x)$, a property of the fitted prior. Proposition~\ref{prop:pull} is the bridge that makes the second predictable from the fitted model alone, before any reconstruction. The prediction is available in closed form because $\mstar = \Real(s\bar x)/|x|^2$ is a \emph{linear} functional of the stacked source at fixed $x$: the posterior of the gain is a scalar mixture with mean $\sum_i \phi_i\,\mathbf{a}^\top\tilde{\mub}_i$, and its variance decomposes by the law of total variance into a within-pair term, the one a scale mixture replaces wholesale, and a between-pair term, which a random scale does not touch. Table~\ref{tab:gate} compares the realised spectral deficit to that prediction on the same 10 excerpts, spectral on both sides so as to isolate the closed-form statement from the synthesis.

The prior accounts for about 70\% of the deficit honestly, 69.6\% at $K = 8$ and 71.4 at $K = 32$, and that share is a quotient of decibels, $6.87/9.87$ and $6.37/8.93$, not a share of error energy. Read in energy the same rows give a smaller figure, $(10^{0.687}-1)/(10^{0.987}-1) = 44.3$\% at $K = 8$, the decibel reading being the compressive one and the one quoted throughout; the residual overconfidence is a factor of two in energy that \emph{decreases} with $K$, the calibration gap falling from 3.00 to 2.72 to 2.56~dB. Three readings follow. The reading of Proposition~\ref{prop:pull} holds on this material: the deficit is within-component variance, 96 to 98\% of the predicted variance sitting in the within-pair term, and no excerpt in 30 rows falls below a ratio of one, the lowest being 1.42. The quartiles keep a thin-tail signature, the model underestimating its own error by a median factor of 15 to 127 on the quarter of bins where its predicted variance is smallest, while on the quarter where that variance is largest the same ratio of realised to predicted error falls to 0.46 and the error is overestimated instead, the tail deficiency of Section~\ref{sec:results:diagnostics} seen through the posterior. And the gap is what a recalibration of \emph{this} fit could recover, 2.56~dB of the 8.93 at $K = 32$ and shrinking with $K$: the decomposition is a definition of the gap on the fitted posteriors, the realised deficit in decibels minus the deficit the predicted variance accounts for, and not a decomposition of error energy, so a model whose predicted variance were exactly calibrated would land on the predicted column and no lower. It bounds nothing about other priors, whose predicted variance can be smaller; what Proposition~\ref{prop:scalemix} adds is that changing the component law to any Gaussian scale mixture transports the whole analysis unchanged, so a scale mixture is pulled onto the line by the same symmetry and its own deficit is again a posterior variance, however that variance is redistributed.

The domain of that statement is measured, not assumed. Decomposing the spectral error without a cross term, $\hat s = \hat m x + q$ and $s = \mstar x + r$ with $q, r \perp x$, gives $\sum|\hat s - s|^2 = \sum|x|^2(\hat m - \mstar)^2 + \sum|q-r|^2$, whose first term is exactly the closed form the table predicts, so the gate is a statement about the in-class part of the error alone. On the same 10 excerpts the time-domain deficits of Table~\ref{tab:deficit} run 0.91 to 0.97~dB above these spectral ones. Since the two columns differ in the reconstruction alone, no retraining and the same bins on both sides, the synthesis is the only candidate mechanism, and the plausible reading is that it falls on the severe side: overlap-add shaves the ceiling residual, whose off-line part $r$ is not consistent with any signal, while the in-class error $(\hat m - \mstar)x$ is itself a real mask and survives it. That reading is not measured here, no arm having isolated the two contributions, so the difference is reported as attributable to the synthesis, not as decomposed. The time-domain numbers quoted everywhere else are therefore the conservative ones.

\subsection{Diagnostics of the Prior}
\label{sec:results:diagnostics}

\begin{table}[tbp]
\renewcommand{\arraystretch}{1.1}
\caption{Distributional diagnostics on isolated source spectra, each against the value of the same statistic on synthetic data satisfying the assumption at the same sample size. Each statistic is to be read against the null immediately to its right, as a ratio, and the reading is given in the text; $K = 32$ for the conditional rows.}
\label{tab:diagnostics}
\centering
\setlength{\tabcolsep}{3pt}
\begin{tabular}{llcccccccc}
\hline
\textbf{Source} & \textbf{Scope} & $\rho$ & null & $\kappa$ & null & $\hat\alpha$ & null & $\gamma$ & null \\
\hline
voc. & marg. & 0.014 & 0.003 & 36.8 & 2.00 & 2.20 & 7.62 & $+0.384$ & 0.000 \\
voc. & comp. & 0.056 & 0.026 & 8.83 & 2.00 & 4.92 & 7.68 & $+0.141$ & 0.000 \\
acc. & marg. & 0.060 & 0.003 & 5.82 & 2.00 & 4.77 & 7.63 & $+0.282$ & 0.000 \\
acc. & comp. & 0.095 & 0.027 & 3.84 & 2.00 & 6.49 & 7.57 & $+0.089$ & $-0.001$ \\
\hline
\end{tabular}
\end{table}

Which assumption to relax is a measurement, and each of the three assumptions of Section~\ref{sec:model:prior} has a statistic that tests it directly on isolated source spectra, before any fit and before any separation. Circularity, $\rho = |\Ex[s^2]|/\Ex[|s|^2]$ per bin, $0$ for a circular law and $1$ for a real one, tests whether there is phase structure to exploit at all. Kurtosis, $\kappa = \Ex[|s|^4]/\Ex[|s|^2]^2$ per bin, $2$ for a circular complex Gaussian, tests the tails, hence where the components are consumed. The Hill estimator $\hat\alpha$ on the upper 5\% of $|s|$, finite for a stable law and unbounded and growing with the sample size for a Gaussian, tests the tail index. And $\gamma = \mathrm{corr}(z_i^2, z_j^2) - \mathrm{corr}(z_i,z_j)^2$ on adjacent bins after rank-gaussianising each marginal tests for a shared latent scale, that is for a scale mixture.

The last one is the discriminating statistic: a bivariate normal of correlation $r$ satisfies $\mathrm{corr}(z_1^2,z_2^2) = r^2$ exactly, so $\gamma$ vanishes for \emph{any} Gaussian copula whatever the marginals are, and the rank transform makes it invariant to any monotone reshaping of them. A positive $\gamma$ therefore cannot come from heavy marginals; it says the dependence itself is non-Gaussian, and a common scale variable is the mechanism that produces it.

Three readings of Table~\ref{tab:diagnostics}, in order of importance. Non-circularity is the weak point and it is what governs the measured deficit: inside the model's own partition, $\rho$ exceeds its noise floor by a factor of 2.2 on vocals and 3.5 on the accompaniment and no more, so the prior has little broken circularity to exploit where it counts, and by Proposition~\ref{prop:pull} that is precisely the quantity that decides whether the posterior mean leaves the line. The median phase rotation under one degree in Table~\ref{tab:readout} is that number seen from the output side. The measurement is made at the volume of the large arm, some 4600 frames per cell against a floor computed at the same sample size, so the small factor is a property of the material at this resolution instead of an artefact of a starved estimate. Tails are the massive violation, $\kappa$ at 37 against 2 marginally and still at 9 conditionally on vocals, with the Hill index well under its floor, and this is where the components go: components devoted to the modulus are not devoted to the geometry, the premise of Proposition~\ref{prop:scalemix} measured, and not assumed. And $\gamma$ stays positive inside the components, so the residual dependence between neighbouring bins escapes any Gaussian copula, the signature of a shared scale.

\subsection{Localisation of the Error Off Support}
\label{sec:results:offsupport}

\begin{figure}[pos=tbp]
\centering
\begin{tikzpicture}[scale=1,>=stealth,line join=round]
  \draw[gray!55] (-0.15,0.000) -- (-0.15,4.608);
  \draw[gray!55] (-0.25,0.576) -- (-0.15,0.576);
  \node[gray!70,anchor=east,font=\scriptsize] at (-0.3,0.576) {-2};
  \draw[gray!55] (-0.25,1.296) -- (-0.15,1.296);
  \node[gray!70,anchor=east,font=\scriptsize] at (-0.3,1.296) {-1};
  \draw[gray!55] (-0.25,2.016) -- (-0.15,2.016);
  \node[gray!70,anchor=east,font=\scriptsize] at (-0.3,2.016) {0};
  \draw[gray!55] (-0.25,2.736) -- (-0.15,2.736);
  \node[gray!70,anchor=east,font=\scriptsize] at (-0.3,2.736) {1};
  \draw[gray!55] (-0.25,3.456) -- (-0.15,3.456);
  \node[gray!70,anchor=east,font=\scriptsize] at (-0.3,3.456) {2};
  \draw[gray!55] (-0.25,4.176) -- (-0.15,4.176);
  \node[gray!70,anchor=east,font=\scriptsize] at (-0.3,4.176) {3};
  \node[gray!70,rotate=90,anchor=south,font=\scriptsize] at (-0.95,2.304) {increase in error (dB)};
  \draw[gray!55] (-0.15,0.000) -- (7.200,0.000);
  \draw[gray!55] (0.000,0.000) -- (0.000,-0.100);
  \node[gray!70,anchor=north,font=\scriptsize] at (0.000,-0.150) {0.0};
  \draw[gray!55] (1.167,0.000) -- (1.167,-0.100);
  \node[gray!70,anchor=north,font=\scriptsize] at (1.167,-0.150) {0.5};
  \draw[gray!55] (2.333,0.000) -- (2.333,-0.100);
  \node[gray!70,anchor=north,font=\scriptsize] at (2.333,-0.150) {1.0};
  \draw[gray!55] (3.500,0.000) -- (3.500,-0.100);
  \node[gray!70,anchor=north,font=\scriptsize] at (3.500,-0.150) {1.5};
  \draw[gray!55] (4.667,0.000) -- (4.667,-0.100);
  \node[gray!70,anchor=north,font=\scriptsize] at (4.667,-0.150) {2.0};
  \draw[gray!55] (5.833,0.000) -- (5.833,-0.100);
  \node[gray!70,anchor=north,font=\scriptsize] at (5.833,-0.150) {2.5};
  \draw[gray!55] (7.000,0.000) -- (7.000,-0.100);
  \node[gray!70,anchor=north,font=\scriptsize] at (7.000,-0.150) {3.0};
  \node[gray!70,anchor=north,font=\scriptsize] at (3.500,-0.550) {tilt parameter $t$};
  \draw[dotted,gray!55] (-0.15,2.016) -- (7.200,2.016);
  \draw[thick] plot coordinates {(0.000,2.016) (0.175,1.838) (0.350,1.687) (0.525,1.548) (0.700,1.416) (0.875,1.287) (1.050,1.160) (1.225,1.035) (1.400,0.912) (1.575,0.790) (1.750,0.670) (1.925,0.552) (2.100,0.436) (2.275,0.323) (2.450,0.218) (2.625,1.188) (2.800,1.197) (2.975,1.116) (3.150,1.038) (3.325,0.963) (3.500,0.892) (3.675,0.825) (3.850,0.827) (4.025,2.813) (4.200,2.800) (4.375,2.788) (4.550,2.777) (4.725,2.768) (4.900,2.761) (5.075,2.757) (5.250,3.664) (5.425,3.700) (5.600,3.699) (5.775,3.699) (5.950,3.701) (6.125,3.704) (6.300,3.708) (6.475,3.713) (6.650,3.719) (6.825,3.726) (7.000,3.734)};
  \fill (2.625,1.188) circle (1.5pt);
  \draw[gray!45,very thin] (2.625,0.000) -- (2.625,1.188);
  \fill (4.025,2.813) circle (1.5pt);
  \draw[gray!45,very thin] (4.025,0.000) -- (4.025,2.813);
  \fill (5.250,3.664) circle (1.5pt);
  \draw[gray!45,very thin] (5.250,0.000) -- (5.250,3.664);
  \node[anchor=east,font=\scriptsize] at (7.000,3.434) {track B, frame 179};
  \draw[thick,gray!60] plot coordinates {(0.000,2.016) (0.175,2.028) (0.350,2.040) (0.525,2.052) (0.700,2.065) (0.875,2.078) (1.050,2.091) (1.225,2.105) (1.400,2.118) (1.575,2.132) (1.750,2.147) (1.925,2.162) (2.100,2.177) (2.275,2.192) (2.450,2.208) (2.625,2.123) (2.800,3.451) (2.975,3.463) (3.150,2.398) (3.325,2.710) (3.500,3.233) (3.675,3.272) (3.850,3.290) (4.025,3.308) (4.200,3.327) (4.375,3.347) (4.550,3.367) (4.725,3.388) (4.900,3.410) (5.075,3.432) (5.250,3.455) (5.425,3.479) (5.600,3.519) (5.775,4.178) (5.950,4.189) (6.125,4.200) (6.300,4.211) (6.475,4.223) (6.650,4.235) (6.825,4.247) (7.000,4.260)};
  \fill (2.800,3.451) circle (1.5pt);
  \draw[gray!45,very thin] (2.800,0.000) -- (2.800,3.451);
  \fill (3.150,2.398) circle (1.5pt);
  \draw[gray!45,very thin] (3.150,0.000) -- (3.150,2.398);
  \fill (3.325,2.710) circle (1.5pt);
  \draw[gray!45,very thin] (3.325,0.000) -- (3.325,2.710);
  \fill (5.775,4.178) circle (1.5pt);
  \draw[gray!45,very thin] (5.775,0.000) -- (5.775,4.178);
  \node[anchor=east,font=\scriptsize] at (7.000,4.560) {track A, frame 135};
  \draw[thick,densely dashed] plot coordinates {(0.000,2.016) (0.175,2.000) (0.350,2.090) (0.525,3.485) (0.700,3.471) (0.875,3.457) (1.050,3.443) (1.225,3.430) (1.400,3.417) (1.575,3.404) (1.750,2.800) (1.925,2.768) (2.100,2.736) (2.275,2.705) (2.450,2.675) (2.625,2.645) (2.800,2.615) (2.975,2.552) (3.150,2.180) (3.325,2.018) (3.500,3.685) (3.675,3.706) (3.850,3.727) (4.025,3.751) (4.200,3.775) (4.375,3.802) (4.550,3.830) (4.725,3.860) (4.900,3.892) (5.075,3.926) (5.250,3.962) (5.425,4.001) (5.600,4.042) (5.775,4.085) (5.950,4.130) (6.125,4.178) (6.300,4.229) (6.475,4.283) (6.650,4.339) (6.825,4.398) (7.000,4.460)};
  \fill (0.525,3.485) circle (1.5pt);
  \draw[gray!45,very thin] (0.525,0.000) -- (0.525,3.485);
  \fill (1.750,2.800) circle (1.5pt);
  \draw[gray!45,very thin] (1.750,0.000) -- (1.750,2.800);
  \fill (2.275,2.705) circle (1.5pt);
  \draw[gray!45,very thin] (2.275,0.000) -- (2.275,2.705);
  \fill (3.325,2.018) circle (1.5pt);
  \draw[gray!45,very thin] (3.325,0.000) -- (3.325,2.018);
  \fill (3.500,3.685) circle (1.5pt);
  \draw[gray!45,very thin] (3.500,0.000) -- (3.500,3.685);
  \node[anchor=east,font=\scriptsize] at (7.000,5.080) {track C, frame 42};
\end{tikzpicture}
\caption{The error of the conditional mean along the spectral tilt, on three of the 36 paths, plotted as the increase over the undeformed frame so that the three are comparable. Markers are the steps at which the arg max of $\boldsymbol{\phi}$ switches, with a light rule dropped to the axis. Between switches the curve moves by hundredths of a decibel per step; at a switch it jumps by about a decibel, in either direction. Track A is \emph{A Classic Education - NightOwl}, B is \emph{ANiMAL - Clinic A}, C is \emph{ANiMAL - Easy Tiger}.}
\label{fig:offsupport}
\end{figure}
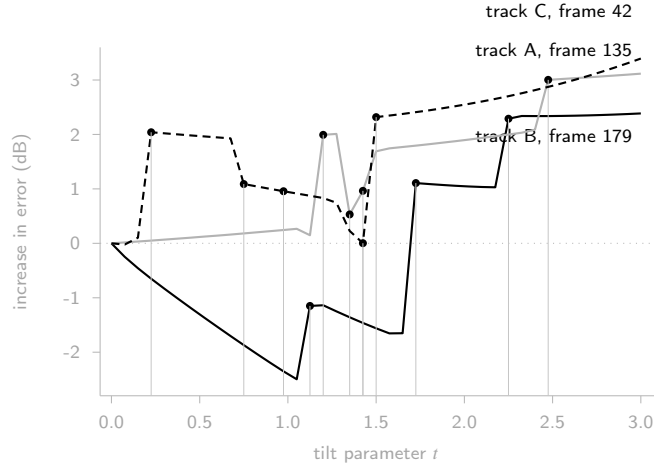

The prediction of Section~\ref{sec:analysis:offsupport} is that the responsibility term of \eqref{eq:twoterm}, and not the regression term, is what fails when the observation leaves the region the prior was fitted on, and it is signed: along a continuous deformation of the observed frame the error should move slowly while the arg max of $\boldsymbol{\phi}$ is stable and break where it switches. The deformation is a one-parameter spectral tilt, $s[f] \mapsto s[f]\exp(t(f/F - 1/2))$ with $t$ swept over 41 values in $[0, 3]$, applied to both sources at once and renormalised per frame and per step so that the observed frame keeps its energy. It is therefore a plain filter, which matters twice: the mixing identity survives, so the truth moves with the observation and the error stays defined, and only the shape of the spectrum changes, where a change of level would leave the support for a reason the analysis does not claim to be about the partition. 12 frames spread over the loudest half of each of the three excerpts give 36 paths, evaluated at $K = 32$ with diagonal covariances at $L = 1024$ on the checkpoints of the in-support arm, nothing being fitted here.

The prediction holds, and not narrowly. Every one of the 36 paths crosses at least one cell boundary, for 209 switches over 1440 steps. The median absolute step in error is 1.00~dB at a switch against 0.03~dB elsewhere, a ratio of 30 on the pooled medians and of 49.7, 32.2 and 16.9 taken excerpt by excerpt, and the extremes are further apart still: the largest step at a switch is 8.50~dB where the 95th percentile of the steps away from switches is 0.35~dB. The refutation, a ratio near one, is nowhere in sight. Figure~\ref{fig:offsupport} shows the shape behind those numbers on three paths, one per excerpt: long stretches on which the estimator tracks the deformation as a single affine map would, punctuated by jumps of about a decibel exactly where the leading component pair changes.

Two qualifications belong with the result. Away from the walls the error drifts instead of climbing: of the 1231 steps at which the arg max is stable, 829 raise the error and 402 lower it, by median amounts of 0.026 and 0.052~dB, and 29 of the 36 paths end above where they started, by a median of 1.99~dB end to end. The fraction of rising steps, 0.86, 0.41 and 0.69 per excerpt, therefore measures nothing about the prediction, since it counts sign changes of hundredths of a decibel; what the deformation produces is a slow upward drift and not a monotone climb. And not every switch breaks the curve. 46 of the 209 move the error by less than the 95th percentile of the ordinary steps, and they are the switches at which the assignment is not sharp: their median $\max\boldsymbol{\phi}$ is 0.85 against 0.997 for the rest. That is the same mechanism seen from the other side, a shared assignment making the estimator a blend of two affine operators, so that swapping which of the two leads moves the estimate very little. The break is a property of a sharp cell boundary and not of the arg max label, which is what \eqref{eq:twoterm} says: the responsibility term is weighted by $\|\boldsymbol{\phi}(\bfx) - \boldsymbol{\phi}(\bfx')\|_1$, and that quantity is small on a boundary the responsibilities cross gradually.

\subsection{Position of the Estimator and Scope of the Measurement}
\label{sec:results:scope}

The ceiling is not the only landmark inside the class, and the second one is unfavourable to the estimator. At $K = 32$ in support the estimator reads 6.40~dB against 13.19~dB for the ideal ratio mask on the same material, so it is 6.78~dB below that mask, the two being paired per excerpt, which is a member of $\Mclass_{[0,1]}$, the smallest of the three nested subclasses of \eqref{eq:class}. Held out at $L = 1024$ the same comparison is wider: the ideal ratio mask reads 13.68~dB pooled against a ceiling of 16.62, so it comes within 2.94~dB of the best real mask on that material, while the smallest deficit of the estimator, 11.44~dB, puts it at 5.18~dB and therefore 8.50~dB below that mask. Those held-out figures are differences of pooled means over the same nine excerpts and not paired per-excerpt differences like the entries of Table~\ref{tab:deficit}, so they state a level and not a significance. About three quarters of the deficit reported above, 8.50~dB of 11.44 held out, is reachable without ever leaving the line, by a better predictor of the real gain alone. What binds the estimator at this level of performance is therefore the prior and not the geometry of the class, and the conflict of Proposition~\ref{prop:pull} becomes the limiting factor only for an estimator that has already come close to the ideal ratio mask, a regime no measurement reported here occupies.

Our measurements support one statement about a class and one about a model, and the two should not be conflated. About the class, the ceiling of $\Mclass_1$ is a data statistic: it is 16.62~dB on the nine retained held-out MUSDB18 excerpts at $L = 1024$, 14.83~dB at $L = 256$ on the five held-out excerpts and 14.88~dB at the same frame length on the five in-support ones, three different sets of material and therefore three separate statements, not a comparison, and no member of the class reaches past the figure of its own configuration, whatever its architecture or training volume. About the model, a closed-form generative prior that is rich enough to leave the class formally, and measurably outside it, sits 9.4 to 12.0~dB below that ceiling held out and 7.6 to 10.9 on the material it was fitted on, over every setting of Tables~\ref{tab:deficit} and~\ref{tab:cov}, and the three knobs that would ordinarily close such a gap, the number of components, the training volume and the covariance structure, account for less than a decibel between them on held-out material. That model also sits well inside the class instead of near its boundary, as the preceding paragraph reports, so the second statement is about a mechanism operating at that level, without claiming that the class has been exhausted. The second statement does not license a claim about generative separation in general, and in particular not about learned generative priors reached by sampling. The estimator studied here is a posterior mean under a Gaussian mixture, and Proposition~\ref{prop:pull} attributes the return to the line to the read-out and the criterion instead of to the family of priors. No learned mask-based separator was run on these excerpts. The ceiling and the subclass landmarks bind such a system by construction, being statistics of the material and not of any estimator, so what is reported here is the position of one estimator against the bounds of its own class and not a ranking of systems; the figures quoted from the literature in the introduction are measured on other material and are not paired with ours.

The objection that such a regime is merely data-limited is the one the two arms were built to settle, and they settle it. The large arm provides 144 to 581 frames per dimension, of the order a rule of thumb would ask for, and the deficit moves by at most 0.19~dB, while the pilot of Section~\ref{sec:protocol:model} shows the estimator to be genuinely sensitive to volume where volume is scarce. The insensitivity reported here is therefore a plateau reached and not assumed, and Section~\ref{sec:analysis:offsupport} explains why the failure of a starved fit lands on the responsibilities instead of on the regression. The open question is now about the family of priors: Proposition~\ref{prop:pull} predicts that a sharper prior lowers the level while the slope in $K$ stays as shallow as measured, a deficit equal to a posterior variance shrinking with the sharpness of the prior and not with the number of components.

Proposition~\ref{prop:pull} read together with Table~\ref{tab:readout} makes the conflict concrete: the read-out that visibly leaves the class is the one that is 1.6~dB worse, by a margin independent of the number of components. The witnesses move toward the line as the fit improves, which is what a deficit equal to a posterior variance predicts and the opposite of what inexact inference would give, and the gate of Section~\ref{sec:results:gate} attributes the bulk of the deficit to the posterior variance of the fitted prior, which is the same finding reached from the other side. Proposition~\ref{prop:pull} turns on the symmetry of the phase posterior and on the criterion instead of on how good the prior is, so it binds a sharper prior in the same way. A separator trained on squared error, or on any criterion whose optimum is a conditional mean, inherits the pull onto the line to the extent that its implicit phase posterior is symmetric about the mixture direction, whatever class its architecture nominally lives in. Escaping the ceiling therefore requires either a prior whose phase posterior is measurably asymmetric, which Table~\ref{tab:diagnostics} shows this material does not supply at the resolution measured, or a criterion whose optimum is not a conditional mean. Which criterion that should be remains open. The established alternatives to squared error in this field, the Itakura-Saito divergence of F\'evotte, Bertin and Durrieu \cite{fevotte2009is} among them, are written on power spectra and are therefore silent about phase instead of free of the pull, while perceptual and adversarial criteria do leave the class by construction and come with no statement about where in the complement of the class they land.

\section{Conclusion}
\label{sec:conclusion}

The class of real gains has an exact ceiling, that ceiling is a statistic of the angular disagreement between source and mixture, and the same projection argument fixes it in one bin, in one frame and over a whole excerpt. Membership in the class, for an estimator that returns a conditional mean, is then decided jointly by the law of the prior and by that read-out. A phase posterior symmetric by reflection about the mixture direction puts the estimate back on the line whatever the prior; non-zero means, non-circularity and inter-frequency coupling are the three assumptions whose relaxation is what lets the estimate leave the line at all, each one term of a chain of nested block structures on stacked spectra; and the excess error of such an estimator is exactly the posterior variance of the oracle gain. Minimising squared error and leaving the class are therefore conflicting requests, and the statement bears on the criterion rather than on a family of priors: it survives replacing every Gaussian component by a Gaussian scale mixture, which makes heavier tails free instead of a rewrite.

The measurements locate a fitted instance of that mechanism instead of exhausting the class. The separator reported here lands 9.4 to 12.0~dB under measured ceilings of 14.8 to 16.6~dB held out, and 7.6 to 10.9 on the material it was fitted on, with no crossing in the 189 distinct per-excerpt figures of Tables~\ref{tab:deficit} and~\ref{tab:cov}, whose 144 and 60 entries share the 15 held-out diagonal measurements at $L = 256$, and which are not independent observations, the same excerpts recurring under every model setting. None of the three knobs that would ordinarily close such a gap does: the number of components, a training volume multiplied by 7.5x and a full covariance account for less than a decibel between them on held-out material, and the last of the three rotates the output 10 times further off the real line for that. A closed-form gate attributes about 70\% of the deficit, as a ratio of decibels, to honest posterior variance, the rest being what recalibrating this fit could recover and nothing more. The measurement is nonetheless taken well inside the class and not at its edge, the ideal ratio mask standing 8.50~dB above the estimator while coming within 2.94~dB of the ceiling, so about three quarters of the deficit is reachable without leaving the line, and the geometry binds only a prior sharper than the one fitted here. What the experiments establish is accordingly the mechanism at a prior-dominated level, without fixing the level itself.

Three limits of the evidence should be carried with those figures. The transport from the spectral analysis to the time-domain ratios reported throughout is an inequality checked numerically rather than proved, so a shortfall is genuine only up to its slack; the time-domain deficits themselves run 0.91 to 0.97~dB above their spectral counterparts, a difference reported as attributable to the synthesis and not decomposed, no arm having isolated the two contributions. The ceilings at the two window lengths are measured on different excerpt sets, so they are separate statements rather than a comparison. And the landmark figures of Section~\ref{sec:results:scope}, unlike the entries of Tables~\ref{tab:deficit} and~\ref{tab:cov}, are differences of pooled means rather than paired per-excerpt differences, and therefore state a level and not a significance.

One construction stands independently of the deficit reported above, and we expect it to be the most immediately usable part of this work. The posterior of Section~\ref{sec:analysis:yardstick} is a genuine two-source separation problem, at full spectral dimension, with an exactly normalised posterior, an exact posterior mean, an exact and observation-independent posterior covariance, and a pointwise log-density. A diffusion or flow-matching sampler can therefore be scored on each of those quantities with no Monte Carlo error on the reference side, which its end metric alone cannot show. The limit of what it provides is the measured realism of the prior, which the diagnostics state in numbers.

Three leads follow from those same diagnostics, in increasing order of what the analysis says they can win. Student components gain heavier tails at one extra parameter each and, by Proposition~\ref{prop:scalemix}, leave the analysis intact beyond restricting the exactness of the reference to a scalar quadrature; they sharpen the posterior where the model is overconfident, and do not remove variance the prior legitimately has. A banded covariance would retain the coupling of $\Mclass_4$ on a neighbourhood of each bin, wherever the frame count can support it, and Section~\ref{sec:results:cov} says what it is for: the unrestricted version does carry the coupling and does leave the line, by an order of magnitude, but at the parameter count of Section~\ref{sec:protocol:model} it obtains its advantage on the material it was fitted on and loses almost all of it elsewhere, so a band is the parametrisation that keeps the first property without incurring the second. And a prior over consecutive frames would exploit the phase advance of a stationary sinusoid between frames, which is deterministic and which \eqref{eq:prior} ignores by construction, the multi-frame filters of \cite{huang2012multiframe,fischer2017mfmvdr} exploiting exactly that dependence from the second-order statistics of the observation instead of from a prior. The last is the lever the calibration gate leaves most room to, since conditioning on the previous frame reduces the within-component variance itself instead of reweighting tails, and that within-pair term is where almost all of the predicted variance sits.

\section*{CRediT authorship contribution statement}

\textbf{Maxime Baelde:} Conceptualization, Methodology, Formal analysis, Software, Investigation, Data curation, Writing -- original draft, Writing -- review and editing, Visualization.

\section*{Declaration of competing interest}

The author declares no known competing financial interests or personal relationships that could have appeared to influence the work reported in this paper.

\section*{Data availability}

The experiments use the public MUSDB18 corpus \cite{musdb18}, and no other data. The code that produces every figure and table of this paper is public, at \url{https://github.com/mbaelde/masking-ceiling}, tag \texttt{v1.0.0}; the separator it measures is the thesis implementation, imported from \url{https://github.com/mbaelde/generative-audio-source-models} rather than reimplemented. The witness statistics quoted in the text and not tabulated here, per-excerpt phase rotations, shares of gains outside $[0,1]$ and the quartile ratios of realised to predicted error, are written out per excerpt and per cell by that code.

\section*{Declaration of generative AI and AI-assisted technologies in the manuscript preparation process}

During the preparation of this work the author used Claude (Anthropic) in order to assist with the drafting of the experimental code, the \LaTeX{} typesetting and the language editing of the manuscript. After using this tool, the author reviewed and edited the content as needed and takes full responsibility for the content of the published article.


\end{document}